\documentclass[preprint,review,12pt,authoryear]{elsarticle}

\usepackage{graphicx}

\usepackage{setspace}
\usepackage{natbib}
\usepackage[cmex10]{amsmath}
\usepackage{amsfonts}
\usepackage{amssymb}
\usepackage{mathrsfs}
\usepackage{amsthm}
\usepackage{bm}
\usepackage{algorithmic}
\usepackage{algorithm}
\usepackage{array}
\usepackage{mdwmath}
\usepackage{mdwtab}
\usepackage[caption=false,font=footnotesize]{subfig}
\usepackage{fixltx2e}
\usepackage{stfloats}
\usepackage{footnote}
\usepackage{tabularx}
\usepackage{multirow}
\usepackage{color}
\usepackage{textcomp}

\newcommand{\Best}{\mathbf {\tilde{A}}}

\newtheorem{proposition}{Proposition}
\newtheorem{theorem}{Theorem}
\newtheorem{lemma}{Lemma}

\newcommand{\R}{{\mathbb R}}
\newcommand{\argmax}{{\rm arg}\max}

\newcommand{\besti}{\tilde{\bm a}^i}

\newcommand{\hist}{{\mathcal H}}

\newcommand{\act}{{\bm a}}
\newcommand{\Act}{{\bm A}}

\begin{document}

\begin{frontmatter}



\begin{spacing}{1.0}

\title{\hbox{Designing Efficient Resource Sharing  For  Impatient Players}  Using Limited Monitoring\tnoteref{thanks}}
\tnotetext[thanks]{This research was supported by National Science Foundation (NSF) Grants No. 0830556, (van der Schaar, Xiao) and  0617027 (Zame) and by the Einaudi Institute for Economics and Finance (Zame).  Any opinions, findings, and conclusions or recommendations expressed in this material are those of the authors and do not necessarily reflect
the views of any funding agency.}

\author[1]{Mihaela van der Schaar}
\address[1]{Department of Electrical Engineering, UCLA. Email: mihaela@ee.ucla.edu.}

\author[2]{Yuanzhang Xiao}
\address[2]{Department of Electrical Engineering, UCLA. Email: yxiao@ee.ucla.edu.}

\author[3]{William Zame}
\address[3]{{\bf Corresponding Author} Department of Economics, UCLA, Los Angeles, CA  90095  \\
Email: zame@econ.ucla.edu; Telephone 310-985-3091 \vspace{-.4in}}

\begin{abstract} {The problem of efficient sharing of a resource is nearly ubiquitous.  Except for pure public goods, each agent's use creates a negative externality; often the negative externality is so strong that efficient sharing is impossible in the short run.  We show that, paradoxically, the impossibility of efficient sharing in the short run enhances the possibility of efficient sharing in the long run, even if outcomes depend stochastically on actions, monitoring is limited and users are not patient.  We base our analysis on the familiar framework of repeated games with imperfect public monitoring, but we extend the framework to view the monitoring structure as chosen by a designer who balances the benefits and costs of more accurate observations and reports.  Our conclusions are much stronger than in the usual folk theorems: we do not require a rich signal structure or patient users and provide an explicit online construction of equilibrium strategies.}
 \end{abstract}
\begin{keyword}
repeated games, imperfect public monitoring, perfect public equilibrium, efficient outcomes, resource allocation games \vspace{.2in}
\JEL C72 \sep C73 \sep D02
%
\end{keyword}

 \end{spacing}

\end{frontmatter}

\pagebreak

\setcounter{page}{1}

\section{Introduction}\label{sec:intro}

The problem of efficient sharing of a resource -- a physical resource, a prize, a market -- is nearly ubiquitous.  Unless the resource is a pure public good,  each agent's use of the resource imposes a negative externality on other users.  Hence (self-interested, strategic) agents will find it difficult to share the resource efficiently, at least in the short run.  In some circumstances -- those we focus on in this paper -- the negative externality is so strong -- competition for the resource is so destructive -- that it will be {\em impossible} for users so share the resource efficiently, at least in the short run.  The purpose of this paper is to show that -- perhaps paradoxically -- the impossibility of efficient sharing in the short run enhances the possibility of efficient sharing in the long run -- even when outcomes depend stochastically on actions, monitoring is very limited and players are not very patient.

We formalize our analysis using the familiar framework of repeated games with imperfect public monitoring but with important differences in both the formulation and the conclusions.  With respect to the formulation, the important difference is in the way we view the monitoring structure.  In the usual models of games with imperfect public monitoring, the monitoring structure is viewed as exogenous and fixed.  In the canonical model of \citet*{GreenPorter1984} for instance, the players compete in a Cournot quantity-setting game but receive feedback only about market prices (rather than quantity choices of other firms) which are determined by random market demand.  We are motivated by the many situations in which the feedback received by the players arises from the action choices of a strategic actor -- the {\em designer} --who  must weigh (among other considerations) the trade-offs between more accurate observations of player actions and more accurate reports provided to players about those observations on the one hand and the costs and consequences of observations and reports on the other hand.  Consider for instance a repeated contest (for details see Example 2 in Section~\ref{sec:example}).  In each period, players choose effort levels which determine (stochastically) the contest winner.  The designer (in this case the contest operator) does not observe effort -- and so certainly cannot announce it -- but does observe the identity of the winner, and could announce that.  However announcing the identify of the winner would violate the privacy of the winner and the losers; if privacy is valued, the designer must weigh the trade-off between the value of maintaining privacy and the (possible) efficiency gain of making more information public.  Because we  wish to emphasize the role played by this and similar {\em choices } of the monitoring structure we formalize an elaborated model in which the choice of the monitoring structure -- both what the designer observes and what the designer announces -- is made explicit.  However, the reduced form that results from our elaborated model once the designer has chosen a monitoring structure looks just the same as the reduced form that is familiar from the standard model of repeated games with imperfect public monitoring.

With respect to conclusions, we cite three important differences: we do not assume a rich signal structure (rather, we require only two signals), we do not assume players are arbitrarily patient  (rather, we find an explicit lower bound on the requisite discount factor), and we provide an {\em explicit (distributed) algorithm} that takes as inputs the parameters -- stage game payoffs, discount factor, target payoff -- and computes the strategy -- the action to be chosen by each player following each public history.  This algorithm can be carried out by each player separately and in real time -- there is no need for the designer to specify/describe the strategies to be played.  A consequence of our constructive algorithm is that the strategies we identify enjoy a useful robustness property: generically, the equilibrium strategies are, for many periods, locally constant in the parameters of the environment and of the problem.

Within our structure, we abstract what we see as the essential features of the resource allocation problems by two assumptions about the stage game.  The first is that for each player $i$ there is a unique action profile $\besti$ that $i$ most prefers.  (In the resource allocation scenario, $\besti$ would be the profile in which only player $i$ accesses the resource.)  The second is that for every action profile $\act$ that is {\em not} in the set $\{\besti\}$ of preferred action profiles the corresponding utility profile $U(\act)$ lies {\em below} the hyperplane $H$ spanned by the utility profiles $\{U(\besti)\}$.  (In the resource allocation scenario, this corresponds to the assumption that allowing access to the resource by more than one individual strictly lowers (weighted) social welfare.)  We capture the  notion that monitoring is very limited by assuming that players do not observe the profile $\act$ of actions but rather only some signal $y \in Y$ whose distribution $\rho(y|\act)$ depends on the true profile $\act$, and that (profitable) single-player deviations from $i$'s preferred action profile $\besti$ can be statistically distinguished from conformity with $\besti$ in the same way.   (But we do not assume that different deviations from $\besti$ can be distinguished from  {\em from each other}.  For further comments, see Examples 2 and 3 in Section 3.)  We emphasize the setting in which there are only two signals -- ``good'' and ``bad'' -- because this setting offers the sharpest results and the clearest intuition and, as we shall see, because two signals are often enough.  To help understand the commonplace nature of our problem and assumptions, we offer three examples: the first is a repeated prisoner's dilemma (although with lower cooperative payoffs than usual), the second is a repeated contest, the third is a repeated resource sharing game.

Not surprisingly we build on the framework of (\citet*{APS1990}; hereafter APS).  Our main technical result (Theorem 1) provides conditions (on the information and payoff structures and the discount factor) that are both  {\em necessary and sufficient} for the set of payoffs  that guarantee each player a given level of security to be self-generating.  Because every payoff vector in a self-generating set can be supported in a perfect public equilibrium (PPE), this leads immediately to sufficient conditions for the same sets to consist of payoff vectors that can be achieved in PPE, and an algorithm for the corresponding PPE strategies (Theorem 2).  Our robustness conclusion (Theorem 3) follows from the nature of the algorithm.  For games with two players,  other considerations lead to the conclusion that maximal sets of PPE payoffs must have a special form and so thus to a characterization of the maximal set of PPE payoffs (Theorem 4).    A surprising aspect of this characterization is that there is a discount factor $\delta^* < 1$ such that any efficient payoff that can be achieved as a PPE payoff for {\em some} discount factor $\delta$ can already be achieved as a PPE payoff as soon as the discount factor $\delta$ exceeds {\em some threshold} $\delta^*$.  Patience is rewarded -- but only up to a point.\footnote{\citet*{MOS2002} establish a similar result for the repeated Prisoner's Dilemma with perfect monitoring;  \citet{AtheyBagwell2001} establish a parallel result for symmetric equilibrium payoffs of two-player symmetric repeated Bertrand games.  We are unaware of any general results that have this flavor.}

The literature on repeated games with imperfect public monitoring is quite large -- much too large to survey here; we refer instead to \citet*{MS2006} and the references therein.  However, explicit comparisons with two papers in this literature may be especially helpful.  The first and most obvious comparison is with  (\citet*{FLM1994}; hereafter FLM) on the Folk Theorem for repeated games with imperfect public monitoring.  As do we, FLM consider a situation in which a single stage game  $G$ with action space $\Act$ and utility function $U : \Act \to \R^n$ is played repeatedly over an infinite horizon; monitoring is public but imperfect, so players do not observe actions but only a public signal of those actions.  In this setting,  ${\rm co}[U(\Act)]$ is the closure of the set of payoff profiles that can be achieved as long run average utilities for {\em some} discount factor and {\em some} infinite set of plays of the stage game $G$.  Under certain assumptions, FLM prove that any payoff vector in the interior of ${\rm co}[U(\Act)]$ that is strictly individually rational can be achieved in a PPE of the infinitely repeated game.    However,  the assumptions FLM maintain are very different from ours in two very important dimensions (and some other dimensions that seem less important, at least for the present discussion).  The first is that the signal structure is rich and informative; in particular, that the number of signals is at least one less than the number of actions of any two players.  The second is that players are arbitrarily patient: that is, the discount factor $\delta$ is as close to 1 as we like.  (More precisely:  given a target utility profile $v$, there is some $\delta(v)$ such that if the discount factor $\delta > \delta(v)$ then there is a PPE of the repeated game that yields the target utility profile $v$.)  In particular, FLM do not identify any PPE for any {\em given} discount factor $\delta < 1$. By contrast, we require only two signals  {\em even if action spaces are infinite} and we do {\em not} assume players are patient:  all target payoffs can be achieved  for some {\em fixed} discount factor -- which may be very far from 1.  Moreover, because FLM consider only payoffs in the interior of ${\rm co}[U(\Act)]$, they have nothing to say about achieving {\em efficient} payoffs.  Their results do imply that efficient payoffs  can be arbitrarily well approximated by payoffs that can be achieved in PPE, but only if the corresponding discount factors are arbitrarily close to 1.  By contrast,  (\citet*{FLT2007}; hereafter FLT) {\em do} show how (some) efficient payoffs can be achieved in PPE.    Given Pareto weights $\lambda_1, \ldots, \lambda_n$ set $\Lambda = \sup \{ \sum \lambda_i U_i(\act) : \act \in \Act \}$ and consider the hyperplane $H = \{x \in \R^n : \sum \lambda_i x_i = \Lambda \}$.  The intersection $H \cap co[U(\Act)]$ is a part of the Pareto boundary of  $ co[U(\Act)]$.  As do we, FLT ask what vectors in  $H \cap co[U(\Act)]$ can be achieved in PPE of the infinitely repeated game.  They identify the largest (compact convex) set $Q \subset H \cap {\rm co}[U(\Act)]$ with the property that every target vector $v \in {\rm int}Q$ (the relative interior of $Q$ with respect to $H$) can be achieved in a PPE of the infinitely repeated game for {\em some} discount factor $\delta(v) < 1$.  However, because FLT consider arbitrary stage games and arbitrary monitoring structures, the set $Q$ identified by FLT may be empty, and FLT do not provide any conditions that guarantee that $Q$ is not empty.  Moreover, as in FLM, FLT assume that players are arbitrarily patient, so do not identify any PPE for any {\em given} discount factor $\delta < 1$.  Having said this, we should also point out that FLT identify the closure of the set of  {\em all} payoff vectors in the interior of $H \cap {\rm co}[U(\Act)]$ that can be achieved in a PPE for some discount factor, while we identify only some.  So there is a trade-off: FLT find more PPE payoffs but provide much less information about the ones they find; we find fewer PPE payoffs but provide much more information about the ones we find.

At the risk of repetition, we want to emphasize the most important features of our results.  The first is that we do not assume discount factors are arbitrarily close to 1.  The importance of this seems obvious in all environments -- especially since the discount factor encodes both the innate patience of players {\em and} the probability that the interaction continues.  The second is that we impose different -- and in many ways weaker -- requirements on the monitoring structure; indeed, we require only two signals, even if action spaces are infinite.  Again, the importance of this seems obvious in all environments, but especially in those in which signals are not generated by some exogenous process  but must be provided by a designer.  In the latter case it seems obvious  -- and in practice may be of supreme importance -- that the designer may wish or need  to choose a simple information structure that employs a small number of signals, saving  on the cost of observing the outcome of play and on the cost of communicating to the agents (and preserving privacy as well).  More generally, the designer may face a trade-off between the efficiency obtainable with a finer information structure and the cost of using that information structure.  (We will return to this point later.)  Finally, because we provide a distributed algorithm for calculating equilibrium play, neither the agents nor a designer need to work out the equilibrium strategies in advance; all calculations can be done online, in real time.

Following this Introduction, Section 2 presents the formal model; Section 3 presents three examples that illustrate the model.  Section 4 presents some preliminary results, presenting conditions under which {\em no} efficient payoffs can be achieved in PPE for {\em any} discount factor.  Section 5 presents the main technical result (Theorem 1); Section 6 presents the implications for PPE (Theorems 2,3) and a comparison with FLT; Section 7 specializes to the case of two players (Theorem 4).  Section 8 returns to the examples to illustrate both the conclusions and the general framework.     Section 9 concludes.  We relegate all proofs to the Appendix.

\section{Model}\label{sec:model}

The reduced form of our model will closely resemble the familiar framework of a repeated game with imperfect public monitoring and we state and prove our formal results in the context of that reduced form.  However,  because we want to emphasize the role played by the designer, we begin by presenting a more elaborated form.

\subsection{Stage Game: Elaborated Form}

There are $n+1$ (potential) actors in our framework: $n$ {\em players} and  a {\em designer}.  Players are characterized by an (exogenously given) {\em game form}:
\begin{itemize}
\item a  (measurable) space $Z$ of {\em outcomes}
\item for each player $i$
\begin{itemize}
\item a  (measurable) space $A_i$ of {\em actions}
\item  a  (measurable) {\em utility function} $u_i : A_i \times Z \to \R$
\end{itemize}
\item a  (measurable) mapping $\bm{a} \mapsto \pi( \cdot | \bm{a}) : {\bm A} = A_1 \times \cdots \times A_n \to \Delta(Z)$
\end{itemize}
We view $\pi(z|\bm{a})$ as the probability that the outcome $z \in Z$ occurs when players choose the action profile $\bm{a} \in \bm{A}$.  Thus the joint actions of players $\bm{a} \in \bm{A}$ stochastically determine an outcome $z \in Z$, and each player's realized utility depends on its own action and the realized outcome.\footnote{We could incorporate actions into the space $Z$ of outcomes so that realized utility depended only on outcomes, but it seems useful to keep separate track of own actions.} For the moment we require only that the spaces $A_i$, the utility functions $u_i$ and the probability mapping $\pi$ be measurable, so that utilities in the reduced form be defined; but later we will insist that the spaces be compact metric and that the  utility functions and the probability mapping be continuous.

The designer is characterized by a  {\em monitoring technology}:
\begin{itemize}
\item a set of $\Phi$ of pairs $(X, \varphi)$ where:
\begin{itemize}
\item  $X$ is a (measurable) space
\item $z \mapsto \varphi(\cdot|z): Z \to \Delta(X)$ is a (measurable) mapping
\end{itemize}
A pair $(X, \varphi)$ is a {\em measurement device}.
\item a set $\Psi$ of pairs $(Y,\psi)$ where
\begin{itemize}
\item $Y$ is a (measurable) space
\item $x \mapsto \psi(\cdot|x): X \to \Delta(Y)$ is a (measurable) mapping
\end{itemize}
A pair $(Y, \psi)$ is an {\em announcement rule}.
\end{itemize}
For the moment, we again require only that the spaces $X, Y$ and the mappings $\varphi, \psi$ be measurable, but later we will insist that the spaces be compact metric and that the mappings be continuous.  Given a choice  $(X,\varphi) \in \Phi$ we interpret $\varphi(x|z)$ as the probability that the designer measures  (observes) $x$ when the outcome $z$ has actually occurred.  Given a choice  $(Y,\psi) \in \Psi$, we interpret $\psi(y|x)$ as the probability that the designer makes the {\em (public) announcement} of the signal $y \in Y$ when the observation $x$ has actually been made.  A pair of choices $(X,\varphi) \in \Phi, (Y,\psi) \in \Psi$ constitute the {\em monitoring structure}.

\subsection{Stage Game: Reduced Form}

The {\em reduced form} of the stage game consists of
\begin{itemize}
\item a set $N = \{1, \ldots, n\}$ of players
\item for each player $i$
\begin{itemize}
\item a (measurable)  space $A_i$ of actions
\item a (measurable) utility function $U_i : \Act= A_1 \times \cdots \times A_n \to \R$
\end{itemize}
\item a (measurable) compact metric space of public signals $Y$
\item a (measurable) map $\act \mapsto \rho( \cdot | \act) : \Act \to \Delta(Y)$
\end{itemize}
We interpret $U_i(\act)$ as $i$'s {\em ex ante} (expected) utility when $\act$ is played and $\rho(y|\act)$ as the probability that the signal $y$ is observed when $\act$ is played.

\subsection{Stage Game: From the Elaborated Form to the Reduced Form}

To pass from the elaborated form to the reduced form we simply define the {\em ex ante} (expected) utilities
$U_i(\act)$ and and the probability distribution $\rho(\cdot| \act)$ over public signals as functions of the action profile $\act$ that is played.  For $\act \in \Act$ and $D \subset Y$ these are:
\begin{eqnarray*}
U_i(\act) &=&  \int_Z u_i(a_i,z) \, d\pi(z|\act) \\
\mbox{} \\
\rho( D | \act) &=& \int_Y \int_X \int_Z {\bm 1}_D \, d\psi(y|x) \, d\varphi(x|z) \, d\pi(z|\act)
\end{eqnarray*}
If $Z, X, Y$ are all finite the last equation can be re-written more simply as
$$
\rho(y|\act) = \sum_{x\in X} \sum_{z\in Z} \psi(y|x) \, \varphi(x|z) \, \pi(z|\act)
$$
Under the maintained assumptions on realized utility, outcome mapping, measurement technology and announcement rules, the derived {\em ex ante} utilities and signal distribution are measurable; if the former are continuous, so are the latter.

\subsection{The Repeated Game with Imperfect Public Monitoring}

In the repeated game, the reduced stage game $G$ is played in every period $t=0,1,2,\ldots$.  Given the signal structure, a {\em public history} of length $t$ is a sequence $(y^0, y^1, \ldots, y^{t-1}) \in Y^t$.  We write $\hist(t)$ for the set of public histories of length $t$,
$\hist^T =  \bigcup_{t=0}^T \hist(t)$ for the set of public histories of length at most $T$ and   $\hist = \bigcup_{t=0}^\infty \hist(t)$   for the set of all public histories of all finite lengths.    A {\em private history} for player $i$  includes  the public history, the actions taken by player $i$, and the realized utilities observed by player $i$,  so a {\em private history} of length $t$ is a  a sequence $(a^0_i, \ldots, a^{t-1}_i; u_i^0, \ldots, u_i^{t-1};y^0, \ldots, y^{t-1}) \in A^t_i \times  \R^t \times Y^t$.  We write $\hist_i(t)$ for the set of $i$'s private histories of length $t$, $\hist_i^T = \bigcup_{t=0}^T \hist_i(t) $ for the set of $i$'s private histories of length at most $T$  and $\hist_i = \bigcup_{t=0}^\infty \hist_i(t)$ for the set of $i$'s private histories of all finite lengths.

A {\em pure strategy} for player $i$ is a mapping from all private  histories into the set of pure
actions $\sigma_i : \hist_i \to A_i$.  A {\em public strategy} for player $i$ is a pure strategy that is independent of $i$'s own action/utility history; equivalently, a mapping from public histories to $i$'s pure actions $\sigma_i : \hist \to A_i$.

We assume all players discount future utilities using the same discount factor $\delta \in (0,1)$ and we use long-run averages, so if the stream of expected utilities is $\{u^t\}$ the vector of long-run average utilities is $(1-\delta) \sum_{t=0}^\infty \delta^{t} u^t$.   A strategy profile $\sigma : \hist_1 \times \ldots \times \hist_n \to \Act$ induces a probability distribution over public and private histories and hence over {\em ex ante} utilities.  We abuse notation and write $U(\sigma)$ for the vector of expected (with respect to this distribution) long-run average {\em ex ante} utilities when players follow the strategy profile $\sigma$.

As usual a strategy profile $\sigma$ is an {\em equilibrium} if each player's strategy is optimal given the strategies of others.  A strategy profile is a {\em public equilibrium} if it is an equilibrium and each player uses a public strategy; it is a {\em perfect public equilibrium (PPE)} if it is a public equilibrium following every public history.

\subsection{Interpretation}

In our formulation, which restricts players to use public strategies, we tacitly assume that players make no use of any information other than that provided by the public signal; in particular, players make no use of information that might be provided by the realized utility they experience each period.  As discussed in \citet*{MS2006}, this assumption admits a number of possible interpretations, each of which is appropriate in some circumstances.  The first is that utility is not realized until the game terminates.  The second is that the outcome $z$ and the public signal $y$ coincide, so that realized utility depends only on own action and the public signal (both of which are observed).  The third is that -- at least in the equilibria and deviations under consideration -- the information provided by realized utility is already provided by the public signal.  (See Example 2 below.)  A fourth is that even if utility {\em is} realized during play {\em and} realized utility {\em does} provide information not provided by the public signal, this additional information is not used.  Lest this last interpretation seems odd, recall that if players other than $i$ follow public strategies then it is optimal for player $i$ to follow a public strategy as well; in particular if other players make no use of information provided by their own realized utility then it is optimal for player $i$ to make no use of information provided by $i$'s realized utility.  (Again, see Example 2 below.)  Finally, it should be kept in mind that by restricting our attention to PPE we are tying our own hands; since our objective is to support efficient sharing, restricting to a particular class of strategies only makes our results stronger.

\subsection{Assumptions on the Stage Game}

To this point we have described a very general setting; we now impose additional assumptions -- first on the stage game and then on the information structure -- that we exploit in our results.

We assume that the spaces $Z, A_i, X, Y$ are all compact metric and that the functions/mappings $u_i, \pi, \varphi, \psi$ are all continuous; as
noted this implies that the functions/mappings $U_i, \rho$ are continuous as well.

Set $U({\bm A}) = \{U({\bm a}) \in \R^n : {\bm a} \in {\bm A}  \}$ and let ${\rm co}(U({\bm A}))$ be the convex hull of $U({\bm A})$.
For each $i$ set
\begin{eqnarray*}
\tilde{v}^i &=& \max_{\act \in \Act} U_i(\act) \\
\besti &=& \argmax_{\act \in \Act} U_i(\act)
\end{eqnarray*}
Compactness of the action space $\Act$ and continuity of utility functions $U_i$ guarantee that $U(\Act)$ and
${\rm co}[U(\Act)]$ are compact, that $\tilde{v}^i$ is well-defined and that the $\argmax$ is not empty.  For convenience, we assume that the $\argmax$ is a singleton; i.e., the maximum utility $\tilde{v}^i$ for player $i$ is attained  at a {\em unique} strategy profile
$\besti$.\footnote{This assumption could be avoided, at the expense of some technical complication.}  We refer to
$\besti$ as $i$'s {\em preferred action profile} and to $\tilde{v}^i = u(\besti)$ as $i$'s {\em preferred utility profile}.  In the context of resource sharing, $\besti$ will typically be the (unique) action profile at which agent $i$ has optimal access to the resource and other agents have none.  For this reason, we will often say that $i$ is {\em active} at the profile $\besti$ and other players are {\em inactive}.  Set $\Best =  \{\besti\}$ and $\tilde{V} = \{\tilde{v}^i\}$ and write $V = {\rm co}\,(\tilde{V})$ for the convex hull of $\tilde{V}$.  Note that ${\rm co}(U(\Act))$ is the closure of the set of vectors that can be achieved -- for {\em some} discount factor -- as long-run average {\em ex ante} utilities of repeated plays of the game $G$ (not necessarily equilibrium plays of course) and that $V$ is the closure of the set of vectors  that can be achieved -- for {\em some} discount factor -- as long-run average {\em ex ante} utilities of repeated plays of the game $G$ in which only actions in $\Best$ are used.  We  refer to ${\rm co}[U(\Act)]$ as the set of {\em feasible  payoffs} and to $V$ as the set of {\em efficient payoffs}.\footnote{The latter is a slight abuse of terminology: because $V$ is the intersection of the set of feasible payoffs with a bounding hyperplane, every payoff vector in $V$ is Pareto efficient and yields maximal weighted social welfare and other feasible payoffs yield lower weighted social welfare -- but other feasible payoffs might also be Pareto efficient.}

We abstract the motivating class of resource allocation problems  by imposing  conditions on the set of preferred utility profiles.  The first is made largely for convenience (and is generically satisfied whenever action spaces are finite); the second abstracts the idea that there are strong negative externalities.

\bigskip

\noindent {\bf Assumption 1 } The vectors $\tilde{v}^1, \ldots, \tilde{v}^n$ are linearly independent.

\bigskip

\noindent {\bf Assumption 2 } The affine span of $\tilde{V}$ is a hyperplane $H$ and all {\em ex ante} utility vectors of the game other than the those in $\tilde{V}$ lie below $H$.  That is, there are weights $\lambda_1, \ldots, \lambda_n > 0$ such that $\sum \lambda_j u_j(\besti) = 1$ for each $i$ and $\sum \lambda_j u_j(\act) < 1$ for each $\act \in \Act, \act \notin \Best$.\footnote{That the sum is 1 is just a normalization.}

\subsection{Assumptions on the Monitoring Structure}

As noted in the Introduction, we focus on the case in which there are only two signals.

\bigskip

\noindent {\bf Assumption 3 }  The set $Y$  contains precisely two signals and $\rho(y|\act) > 0$ for every $y \in Y$ and $\act \in \Act$.  (The
monitoring structure has {\em full support}.)

\bigskip

We assume that profitable deviations from the profiles $\besti$ exist and  be statistically detected in a particularly simple way.

\bigskip

\noindent {\bf Assumption 4 }  For each $i \in N$ and each $j \not= i$ there is an action $a_j \in A_j$ such that $u_j(a_j, \besti_{-j}) >
u_j(\besti)$.  Moreover, there is a labeling $Y = \{y^i_g, y^i_b\}$ with the property that
$$
a_j \in A_j, U_j(a_j, \besti_{-j}) > U_j(\besti) \Rightarrow \rho(y^i_g| a_j, \besti_{-j}) < \rho(y^i_g |, \besti)
$$
That is, given that other players are following $\besti$, any strictly profitable deviation by player $j$ strictly reduces the probability that the
``good'' signal $y^i_g$ is observed (equivalently: strictly increases the probability that the ``bad'' signal $y^i_b$ is observed).

\bigskip

 The import of Assumption 4 is that all profitable single player deviations from $\besti$ alter the signal distribution in the {\em  same direction} although  perhaps not to the same extent.  We allow for the possibility that non-profitable deviations may not be detectable in the same way -- perhaps not detectable at all -- and for the possibility that which signal is ``good'' and which is ``bad''  depend on the identity of the {\em active} player $i$.

\section{Examples}\label{sec:example}

The assumptions we have made -- about the structure of the game and about the information structure -- are far from innocuous, but they apply in a wide variety of interesting environments.  Here we describe three simple examples which motivate and illustrate the assumptions we have made and the conclusions to follow.  We present the first example directly in the reduced form and the other two examples in both the elaborated and reduced forms.

\bigskip

\noindent {\bf Example 1: A Repeated Prisoners' Dilemma}

We begin by discussing a simple Prisoner's Dilemma but with a payoff structure slightly different from the familiar one; see Table~\ref{table:PrisonersDilemmaGame}.   For our purposes we assume $B > 2c > 2b > 0$.  As usual, $(D,D)$ is a strictly dominant strategy profile; the difference between the payoffs shown here and the usual ones is that $(C,C)$ is Pareto dominated by randomizing between $(C,D)$ and $(D,C)$.  See Figure \ref{fig:PayoffRegion_PrisonersDilemma}.

There are two signals: $Y = \{y_g, y_b\}$; the probability distribution over signals following actions is
\begin{eqnarray}
\pi(y_g |\bm{a}) = \left\{\begin{array}{cl} p & \mbox{if} \ \ \  \bm{a}=\mathrm{(C,C)} \\ q & \mbox{if} \ \ \  \bm{a}=\mathrm{(C,D)~or~(D,C)} \\ r & \mbox{if} \ \ \   \bm{a}=\mathrm{(D,D)}
\end{array}\right.
\end{eqnarray}
where $p, q, r \in (0,1)$; for our purposes we assume $p \geq q > r$.  It is easily checked that the stage game and monitoring structure satisfy our assumptions.  (Note that $y_g$ is the good signal for both players.)  As we will show in Section 4, we can completely characterize the most efficient outcomes that can be achieved in a PPE.  To summarize the conclusion, for each discount factor $\delta \in (0,1)$ write $E(\delta)$ for the set of efficient (average) payoffs that can be achieved when the discount factor is $\delta$.  Set
$$
\delta^* \ =  \ \frac{1}{1+\left(\frac{B-2\frac{q}{q-r}b}{B+2\frac{1-q}{q-r}b}\right)}
$$
It follows from Theorem 4 that if $\delta \geq \delta^*$ then
$$
E(\delta) = \{(v_1, v_2) : v_1 + v_2 = B; v_i \geq q/(q-r) b \}
$$
Note that the set of efficient equilibrium outcomes {\em does not} increase as $\delta \to 1$; as we noted in the Introduction, patience is rewarded but only up to a point.  See Figure \ref{fig:PayoffRegion_PrisonersDilemma_PPE}.

\begin{table}
\renewcommand{\arraystretch}{1.1}
\caption{Modified Prisoners' Dilemma} \label{table:PrisonersDilemmaGame} \centering
\begin{tabular}{r|c|c|}
\multicolumn{1}{r}{}
 &  \multicolumn{1}{c}{C}
 & \multicolumn{1}{c}{D} \\
\cline{2-3}
C & $(c, c)$ & $(0, B)$ \\
\cline{2-3}
D & $(B, 0)$ & $(b, b)$ \\
\cline{2-3}
\end{tabular}
\end{table}

\begin{figure}
\centering
\includegraphics[width =3.0in]{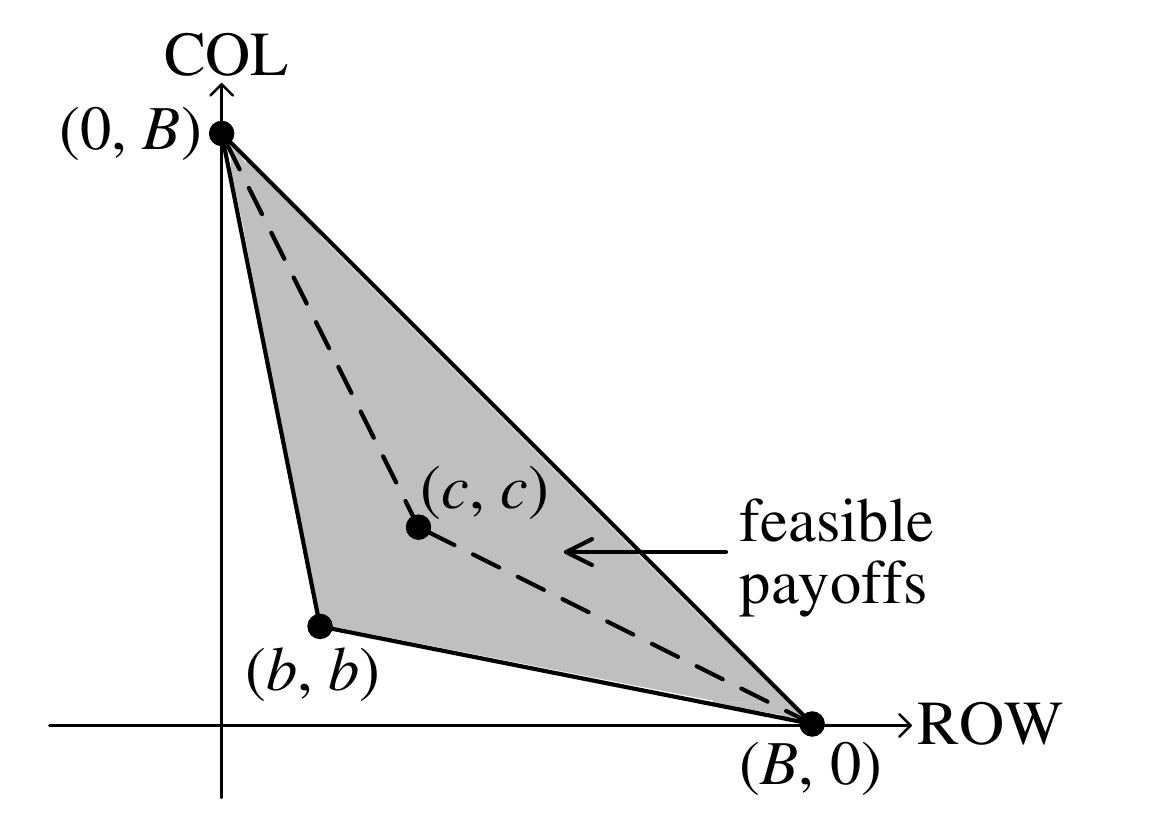}
\caption{Feasible Region for the Modified Prisoners' Dilemma\label{fig:PayoffRegion_PrisonersDilemma}}
\end{figure}

\bigskip

\noindent {\bf Example 2: A Repeated Contest }

We consider a repeated contest.  In each period, a set of $n \geq 2$ players competes for the use of a single indivisible resource/prize each of them values at $R > 0$.   Winning the contest depends (stochastically) on the effort exerted by each player; we write $A_i = [0,1]$ for the set of $i$'s effort levels (actions).  Each agent's effort interferes with the effort of others and there is always  some probability that no one wins (the prize is not awarded) independently of the choice of effort levels.  If $\act = (a_i)$ is the vector of effort levels then the probability agent $i$ obtains the wins the contest (obtains the resource/prize) is
$$
{\rm Prob}(i \ {\rm wins} | \act) = a_i \left( \eta - \kappa \sum_{j\not=i} a_j \right)^+
$$
where $\eta, \kappa \in (0,1)$ are parameters.   The assumption that $\eta < 1$ reflects that there is always some probability the prize is not awarded; $\kappa$ measures the strength of the interference.  Notice that competition is destructive: if more than one agent exerts effort that lowers the probability that {\em anyone} wins the prize.  Utility is separable in reward and effort; effort is costly with constant marginal cost $c > 0$.  To avoid trivialities and conform with Assumptions 1-4 we assume $R\eta > c$ and that
$\kappa >\frac{1}{2}\left(\eta-\frac{c}{R}\right)$.

In the elaborated form  of the stage game, players are $N = \{1, \ldots, n\}$, action sets are $A_i = [0,1]$, outcomes are $Z = \{z_0, \ldots, z_n \}$ (where $z_0$ is interpreted as ``no one wins'' and $z_i$ is interpreted as ``$i$ wins'') and $i$'s realized utility as a function of his own effort level $a_i$ and the outcome $z$ is
$$
u_i(a_i, z_k) = \left\{ \begin{array}{rcl}
                              R - ca_i & \mbox{ if } k = i \\
                              - c a_i & \mbox{ if } k \not= i
                              \end{array}
                              \right.
$$
In this context it seems natural to assume that the designer observes who wins -- how else could the prize be awarded? -- so that $X = Z$ and $\varphi$ is the identity.  We assume that the designer wishes to preserve privacy so announces only whether or not {\em some} player won the contest but not {\em the identity of the winner}.  Hence the reporting rule $(Y^1, \psi^1)$, $(Y^2, \psi^2),$ where
\begin{itemize}
\item $Y^1 = \{y_b, y_g\}$; $\psi^1(z_k) = y_b$ if $k = 0$,  $\psi^1(z_k) = y_g$ if $k \not= 0$
\item $Y^2  = Z$; $\psi^2(z_k) = z_k$ for all $k = 0, \ldots, n$
\end{itemize}
In the first case, the designer announces whether or not there has been a winner; in the second case the designer also announces the identity of the winner.

In the reduced forms of the stage game, the {\em ex ante} expected utilities  are given by
\begin{eqnarray*}
U_i(\act) &=& a_i \left( \eta - \kappa \sum_{j \not= i} a_j \right)^+ R - ca_i \\
\end{eqnarray*}
In the first case,  the signal distribution is
$$
\overline{\rho}(y_* | a) \left\{ \begin{array}{rcl}
                            1 - \sum_i a_i \left(\eta - \kappa \sum_{j \not= i} a_j \right)^+
                                  &\mbox{ if } & * = b \\

                            \sum_i a_i \left(\eta - \kappa \sum_{j \not= i} a_j \right)^+ &\mbox{ if } & *= g
                            \end{array} \right.
$$
In the second case the signal distribution is
$$
\widehat{\rho}(y_k | a) = \left\{ \begin{array}{rcl}
                            1 - \sum_i a_i \left(\eta - \kappa \sum_{j \not=i} a_j \right)^+
                                  &\mbox{ if } & k = 0 \\
                            a_k \left(\eta - \kappa \sum_{j \not=k} a_j \right)^+ &\mbox{ if } & k \not= 0
                            \end{array} \right.
$$
Straightforward but somewhat messy calculations show that in either case the reduced form satisfies all of our assumptions.  (Player $i$'s preferred action profile $\besti$ has $\besti_i = 1$ and $\besti_j = 0$ for $j \not= i$: $i$ exerts maximum effort, others exert none.  Note that this does not guarantee that $i$ wins the contest -- there may still be no winner -- but the effort profiles $\besti$ are precisely those that maximize the probability that {\em someone} wins the prize.)

The first reporting rule preserves privacy, the second rule does not.  However, the second reporting rule provides more information to players.  Suppose for instance that a strategy profile $\sigma$ calls for $\besti$ to be played after a particular history.   If all players follow $\sigma$ then only player $i$ exerts non-zero effort so  only two outcomes can occur: either player $i$ wins or no one wins.  If player $j \not= i$ deviates by exerting non-zero effort, a third outcome can occur:  $j$ wins.  With either monitoring structure, it is possible for the players to detect (statistically) that {\em someone} has deviated -- the probability that {\em someone} wins goes down -- but with the second monitoring structure it is also possible for the players to detect (statistically) {\em who} has deviated -- because the probability that the deviator wins becomes positive.  Hence, with the first monitoring structure all deviations must be ``punished'' in the same way, but with the second monitoring structure, ``punishments'' can be tailored to the deviator.  If punishments can be ``tailored'' to the deviator then punishments can be more severe; if punishments can be more severe it may be possible to sustain a wider range of PPE.  Which reporting rule -- hence which monitoring structure -- should be chosen by the designer will depend on the tradeoff the designer makes between preserving privacy and sustaining a wider range of PPE.  We will see a similar but even starker tradeoff in Example 3 following.

\bigskip

\noindent {\bf Example 3: Resource Sharing}

We consider $n \geq 3$ users (players) who send information packets through a common server.  The server has a nominal capacity of $\chi > 0$ (packets per unit time) but the capacity is subject to random shocks so the actually realized capacity in a given period is $\chi - \varepsilon$, where the random shock $\varepsilon$ is distributed in some interval $[0,\bar{\varepsilon}]$ with (known) distribution $\nu$.
  In each period, each player chooses a packet rate (packets per unit time) $a_i \in A_i = [0,\chi]$.  This is a well-studied problem; assuming that the players' packets arrive according to a Poisson process, the whole system can be viewed as what is known as an  {\em M/M/1 queue}; see \citet*{FlowControl} for instance.  It follows from the standard analysis that if
$\varepsilon$ is the realization of the shock then packet deliveries will be be subject to a {\em delay} of
$$
d(\bm{a},\varepsilon) = \left\{\begin{array}{lll} 1/(\chi-\varepsilon-\sum_{i=1}^n a_i)  & \mathrm{ if } & \sum_{i=1}^n a_i<\chi-\varepsilon \\
\infty & \mathrm{if} &  \sum_{i=1}^n a_i  \geq \chi-\varepsilon
\end{array}\right.
$$
Given the delay $d $, each player's realized utility is its ``power'', namely the ratio of the $p$-th power of its own packet rate to the delay:
$$
u_i(\bm{a},d) = a_i^p/d
$$
where $p>0$ is a parameter that represents trade-off between rate and delay.\footnote{In order to guarantee that the reduced form satisfies our assumptions we assume $\bar{\varepsilon}\leq\frac{2}{2+p}\chi$.}  (If delay is infinite utility is 0.)  Formally, we identify the
outcome with the pair consisting of the vector $\bm{a}$ of packet rates and the realized shock $\varepsilon$, so $Z = A \times [0,\bar{\varepsilon}]$
and $\pi( \cdot | \bm{a}) = \delta_{\bm{a}} \times \nu$ where $\delta_{\bm a}$ is point mass at $\bm{a}$ and $\nu $ is the given distribution
of shocks.

The designer does not observe packet rates but can measure the delay, but with error and at a cost.  Thus the space of measurements is $X =  [0,\infty]$ and the measurement technology consists of a space of maps $(\act,\varepsilon) \mapsto
\varphi(\cdot | (\act,\varepsilon)): \Act \times [0,\bar{\varepsilon}] \to \Delta(X)$.  Many possible reporting technologies are possible; we assume the designer reports only whether the measured delay  was  above or below a chosen threshold $d_0$; say $Y = \{y_\ell, y_h\}$ where $y_\ell$ is interpreted as ``delay was low (below $d_0$)'' and  $y_h$ is interpreted as ``delay was high (above $d_0$).''

In the reduced form, each  player $i$'s ex-ante payoff is
\begin{eqnarray}
U_i(\bm{a}) &=& \left\{\begin{array}{cl} a_i^p \, (\chi-\frac{\bar{\varepsilon}}{2}-\sum_{j=1}^n a_j) & \mathrm{ if } \  \ \sum_{j=1}^n a_j\leq\chi-\bar{\varepsilon} \\
a_i^p \, (\chi-\sum_{j=1}^n a_j) \, \frac{\chi-\sum_{j=1}^n a_j}{2\bar{\varepsilon}} & \mathrm{ if }\ \  \chi-\bar{\varepsilon}<\sum_{j=1}^n a_j<\chi \\
0 & \mathrm{ otherwise }
\end{array}\right. \nonumber
\end{eqnarray}
and the distribution of signals is
$$
\rho(y_\ell| \act) =\int_{0}^{\chi-\sum_{j=1}^n a_j-\frac{1}{d_0}} d\,\nu(x) = \frac{[\chi-\sum_{j=1}^n
a_j-\frac{1}{d_0}]_{0}^{\bar{\varepsilon}}}{\bar{\varepsilon}},
$$
where $[x]_a^b\triangleq\min\{\max\{x,a\},b\}$ is the projection of $x$ in the interval $[a,b]$.
Note that $y_\ell$ is the ``good'' signal: deviation from any preferred action profile increases the probability of realized delay, hence increases the probability of measured delay, and reduces the probability that reported delay will be below the chosen threshold.

It might seem to the reader that the players could back out realized delay from their own realized utility and hence that announcements are irrelevant -- but this is not quite so.  Players who choose packet rates greater than 0 {\em can} back out realized delay from their own realized utility but at any one of the preferred action profiles $\besti$ and at any single-player deviation from any one of the preferred action profiles $\besti$, at least one player will choose a packet rate $a_j = 0$ and hence will experience realized utility $U_i(\act) = 0$; that player {\em cannot}  back out observed delay.  Hence announcements serve to (statistically) inform players {\em who have complied} of the existence of some player {\em who has not complied}.  Put differently, announcements serve to keep all players on the same informational page.

\section{Ruling out Some Efficient PPE Payoffs}

Throughout this Section, we consider a fixed reduced form and maintain the notation and assumptions of Section 2.  Our ultimate goal is to find conditions -- on the discount factor among other things -- that enable us to construct PPE that achieve payoffs in $V$ (efficient payoffs).

We first show that under certain conditions, certain efficient payoffs {\em cannot} be achieved in PPE no matter what the discount factor is.  To this end, we identify   two measures of benefits from deviation.  (These same measures will play a prominent role in the next Section as well.)  Given $i, j \in N$ with $i \not= j$ set:
\begin{eqnarray}
\alpha(i,j) &=&
 \sup \Big\{  \frac{u_j(a_j, \besti_{-j}) - u_j( \besti)}{\rho(y_b^i|a_j, \besti_{-j})-\rho(y_b^i| \besti)} : \nonumber \\
&& \hspace{0.3in}  a_j\in A_j, u_j(a_j,\besti_{-j})>u_j( \besti) \Big\} \\
\beta(i,j) &=&
 \inf \Big\{  \frac{u_j(a_j, \besti_{-j}) - u_j( \besti)}{\rho(y_b^i|a_j, \besti_{-j})-\rho(y_b^i| \besti)} : \nonumber \\
&& \hspace{0.3in}  a_j\in A_j, u_j(a_j,\besti_{-j}) < u_j( \besti), \rho(y_b^i|a_j, \besti_{-j})<\rho(y_b^i| \besti) \Big\}
\end{eqnarray}
(We follow the usual convention that the supremum of the empty set is $-\infty$ and the infimum of the empty set is $+\infty$.)

Note that $u_j(a_j, \besti_{-j}) - u_j( \besti)$ is the gain or loss to player $j$ from deviating from $i$'s preferred action profile $\besti$ and $\rho(y_b^i|a_j, \besti_{-j})-\rho(y_b^i| \besti)$ is the increase or decrease in the probability that the bad signal occurs (equivalently, the decrease or increase in the probability that the good signal occurs) following the same deviation.  In the definition of $\alpha(i,j)$  we consider only deviations that are strictly profitable; by assumption, such deviations strictly increase the probability that the bad signal occurs, so $\alpha(i,j)$ is either $-\infty$ or strictly positive.  In the definition of $\beta(i,j)$ we consider only deviations that are strictly unprofitable {\em and} strictly decrease the probability that the bad signal occurs, so $\beta(i,j)$ is the infimum of strictly positive numbers and so is necessarily $+\infty$ or finite and non-negative.\footnote{Note that if we strengthened Assumption 4 so that {\em any} deviation -- profitable or not -- increased the probability of a bad signal (as is the case in Examples 1-3 and would be the case in most resource allocation scenarios), then $\beta(i,j)$ would be the infimum of the empty set whence $\beta(i,j) = + \infty$.}

To understand the significance of these numbers, think about how player  $j$ could gain by deviating from $\besti$.  Most obviously, $j$ could gain by deviating to an action that {\em increases} its current payoff.  By assumption, such a deviation will {\em increase} the probability of a bad signal; assuming that a bad signal leads to a lower continuation utility, whether such a deviation will be profitable will depend on the current gain and on the change in probability;  $\alpha(i,j)$ represents a measure of net profitability from such deviations.  However, player $j$ could also gain by deviating to an action that {\em decreases} its current payoff but also {\em decreases} the probability of a bad signal, and hence leads to a higher continuation utility.   $\beta(i,j)$ represents a measure  of net profitability from such deviations.

Because $\tilde{V}$ lies in the supporting hyperplane $H$ and the utilities for action profiles not in $\Best$ lie strictly below $H$, in order that the strategy profile $\sigma$ achieves an efficient payoff it is necessary and sufficient that $\sigma$ use only preferred action profiles:
$U(\sigma) \in V$ if and only if $\sigma(h) \in \Best$ for every public history $h$ (independently of the discount factor $\delta$).  For PPE strategies we can say a lot more.  The first Proposition is almost obvious; the second and third seem far from obvious.  (All proofs are in the Appendix.)

\bigskip

\begin{proposition}  In order that $\tilde{v}^i$ be achievable in a PPE equilibrium (for any discount factor $\delta$) it is necessary and sufficient that $u_j(a_j , \besti_{-j}) \leq u_j(\besti)$ for every $j \not=i$ and every $a_j \in A_j$.
\end{proposition}

\bigskip

\begin{proposition}\label{prop:deviation-j}  If $\sigma$ is an efficient PPE (for any discount factor $\delta$) and $i$ is active following some history (i.e., $\sigma(h) = \besti$ for some $h$) then
\begin{eqnarray}
\alpha(i,j)\leq \beta(i,j)
\end{eqnarray}
 for every $j \in N, j \not=i$.
\end{proposition}

\bigskip

\begin{proposition}\label{prop:deviation-i}  If $\sigma$ is an efficient PPE (for any discount factor $\delta$) and $i$ is active following some history (i.e., $\sigma(h) = \besti$ for some $h$) then
\begin{eqnarray}
\tilde{v}_i^i-u_i(a_i,\bm{\tilde{a}}_{-i}^i) \geq \frac{1}{\lambda_i} \, \sum_{j\neq i} \lambda_j \,
\alpha(i,j) \left[ \rho(y_b^i|a_i,\besti_{-i})-\rho(y_b^i|\besti) \right]
\end{eqnarray}
\end{proposition}

\bigskip

The import of Propositions 2 and 3 is that if any of these inequalities fail then certain efficient payoff vectors can {\em never} be achieved in PPE, no matter what the discount factor is.  In the next Sections, we show how these inequalities and other conditions yield necessary and sufficient conditions that certain sets be self-generating and hence yield sufficient conditions for efficient PPE.

Proposition 2 might seem quite mysterious: $\alpha$ is a measure of the current gain to deviation and $\beta$ is a measure of the future gain to deviation; there seems no obvious reason why PPE should necessitate any particular relationship between $\alpha$ and $\beta$.  As the proof will show, however, the assumption of two signals and the efficiency of payoffs in $V$ imply that $\alpha$ is bounded above and $\beta$ is bounded below by the same quantity, which is a weighted difference of continuation values -- a quantity that does have an obvious connection to PPE.

\section{Characterizing Efficient Self-Generating Sets}

As in the previous Section, we consider a fixed reduced form and maintain the notation and assumptions of Section 2.  In order to find efficient PPE payoffs we follow APS and look for self-generating sets of efficient payoffs.

Fix a subset $W \subset {\rm co}[U(\Act)]$ and a {\em target payoff} $v \in {\rm co}[U(\Act)]$.  Recall from APS that $v$ can be {\em decomposed with respect to} $W$ (for a given discount factor $\delta < 1$) if there exist an action profile $\act \in \Act$ and continuation payoffs $\gamma: Y \to W$ such that
\begin{itemize}
\item $v$ is the (weighted) average of current and continuation payoffs when players follow $\act$
$$ v = (1-\delta)U(\act) + \delta \sum_{y \in Y}\rho(y|\act) \gamma(y) $$
\item continuation payoffs provide no incentive to deviate: for each $j$ and each $a_j \in A_j$
$$
v_j \geq (1-\delta)U(a_j, \act_{-j}) + \delta \sum_{y \in Y}\rho(y|a_j, \act_{-j}) \gamma(y)
$$
\end{itemize}
 Write ${\mathcal B}(W,\delta)$ for the set of target payoffs $v \in {\rm co}[U(\Act)]$ that can be decomposed with respect to $W$ (for the discount factor $\delta$.  Recall that $W$ is {\em self-generating} if $W \subset {\mathcal B}(W,\delta)$; i.e., every target vector in $W$ can be decomposed with respect to $W$.

 Because $V$ lies in the hyperplane $H$, if $v \in V$ and it is possible to decompose  $v \in V$ with respect to {\em any} set and for {\em any} discount factor, then the associated action profile $\act$ must lie in  $\Best$ and the continuation payoffs must lie in $V$.  Because we are interested in efficient payoffs we can therefore restrict our search for self-generating sets to subsets $W \subset V$.     In order to understand which sets $W \subset V$ can be self-generating, we need to understand how players might profitably gain from deviating from the current recommended action profile.  Because we are interested in subsets $W \subset V$, the current recommended action profile will always be $\besti$ for some $i$, so we need to ask  how a player $j$ might profitably gain from deviating from
$\besti$.  For player $j \not= i$, a profitable deviation might occur in one of two ways: $j$ might gain by choosing an action $a_j \not= \besti_j$ that increases $j$'s {\em current} payoff or by choosing an action $a_j \not= \besti_j$ that alters the signal distribution in such a way as to increase $j$'s {\em future} payoff.  Because $\besti$ yields $i$ its best current payoff, a profitable deviation by $i$ might occur only by choosing an action that  that alters the signal distribution in such a way as to increase $i$'s {\em future} payoff.  In all cases, the issue will be the net of the current gain/loss against the future loss/gain.

We focus attention on sets of the form
$$
V_\mu = \{v \in V: v_i \geq \mu_i \mbox{ for each } i \}
$$
where $\mu \in \R^n$; we assume without further comment that $V_\mu \not=\emptyset$.  For lack of a better term, we say that $V_\mu$ is {\em regular} if for each $i \in N$ there is a vector $\hat{v}^i \in V_\mu$ such that $\hat{v}^i_j = \mu_j$ for each $j \not= i$.  Whether or not $V_\mu$ is regular depends both on the shape of $V$ and on the magnitude of $\mu$: see Figures  \ref{figure:V-mu1}, \ref{figure:V-mu2}, \ref{figure:V-mu3} for instance.  A few simple facts are useful to note:
\begin{itemize}
\item If $\tilde{v}^i_j = 0$ for all $i, j \in N$ with $i \not= j$ (as is the case in many resource sharing scenarios such as Examples 2, 3) then $V_\mu$ is regular for every $\mu \geq 0$.
\item If $V_\mu \not= \emptyset$ and $V_\mu$ is a subset of the interior of $V$ (relative to the hyperplane $H$) then $V_\mu$ is regular.
\item If $v$ lies in the interior of $V$ (relative to the hyperplane $H$) and $\mu = v - \epsilon \cdot {\bm 1}$ for $\epsilon > 0$ sufficiently small, then $v \in V_\mu$ and  $V_\mu$ is regular.
\item If $V_\mu$ is not a singleton then it must contain a point of the interior of $V$ (relative to the hyperplane $H$).
\end{itemize}
If $V_\mu$ is a singleton, it can only be a self-generating set (and hence achievable in a PPE) if $V_\mu =  \tilde{v}^i$ for $i$; because we have  already characterized this possibility in Proposition 1, we focus on the non-degenerate case in which $V_\mu$ is not a singleton and hence contains a point of the interior of $V$.  Note that a point in the interior of $V$ can only be achieved by a repeated game strategy in which  {\em all} players are active following some history.

\begin{figure}
\centering
\includegraphics[width =3.0in]{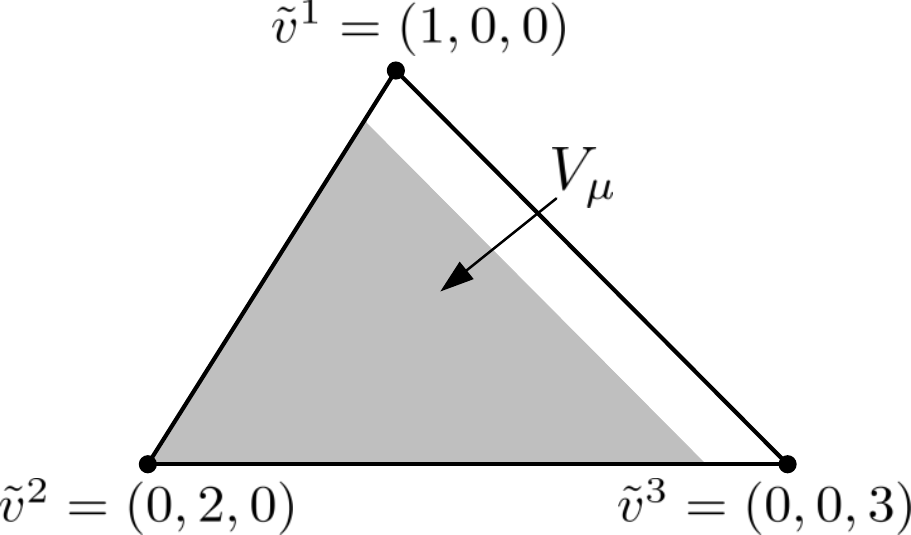}
\caption{$\mu = (0,1/4,0)$; $V_\mu$ is regular}\label{figure:V-mu1}
\end{figure}

\begin{figure}
\centering
\includegraphics[width =3.0in]{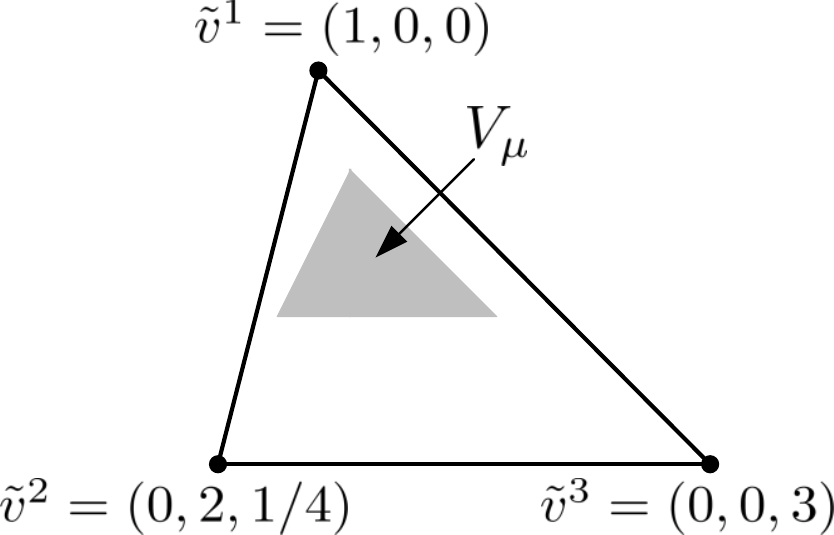}
\caption{$\mu = (1/2,1/2,1/2)$; $V_\mu$ is regular}\label{figure:V-mu2}
\end{figure}

\begin{figure}
\centering
\includegraphics[width =3.0in]{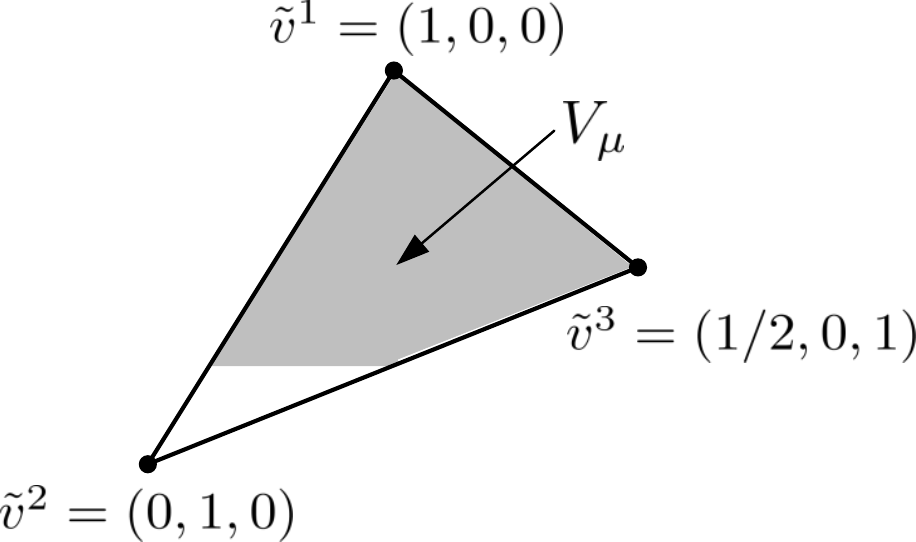}
\caption{$\mu = (1/4,0,0)$; $V_\mu$ is not regular}\label{figure:V-mu3}
\end{figure}

The following result provides necessary and sufficient conditions on $\mu$, the payoff structure, the information structure and the discount factor that a regular $V_\mu$ be a self-generating set.


\bigskip

\begin{theorem}\label{theorem:CharacterizeEquilibriumPayoff}  Fix $\mu$; assume that $V_\mu$ is regular and not an extreme point of $V$.   In order that $V_\mu$ be a self-generating set, it is necessary and sufficient that the following conditions be satisfied:
\begin{itemize}
\item[{}] {\bf Condition~1 } for all $i, j \in N$ with $i \not= j$:
\begin{eqnarray}
\alpha(i,j) \leq \beta(i,j)
\end{eqnarray}
\item[{}]  {\bf Condition~2 } for all $i \in N$ and all $a_i \in A_i$:
\begin{eqnarray}
\tilde{v}_i^i-u_i(a_i,\bm{\tilde{a}}_{-i}^i) \geq \frac{1}{\lambda_i} \, \sum_{j\neq i} \lambda_j \,
\alpha(i,j) \left[ \rho(y_b^i|a_i,\besti_{-i})-\rho(y_b^i|\besti) \right]
\end{eqnarray}
\item[{}]  {\bf Condition~3 } for all $ i \in N$:
\begin{eqnarray}
\mu_i \geq \max_{j\neq i} \left( \tilde{v}_i^j + \alpha(j,i)[1-\rho(y_b^j|\bm{\tilde{a}}^j)] \right)
\end{eqnarray}
\item[{}] {\bf  Condition~4 } the discount factor $\delta$ satisfies:
\begin{eqnarray}
\delta \geq \underline{\delta}_{\mu} \triangleq \left(1+\frac{1-\sum\limits_{i} \lambda_i\mu_i}{\sum\limits_{i} \left[  \lambda_i\tilde{v}_i^i +  \sum\limits_{j\neq i}
\lambda_j \, \alpha(i,j) \, \rho(y_b^i|\bm{\tilde{a}}^i) \right]  - 1} \right) ^{-1}
\end{eqnarray}
\end{itemize}
\end{theorem}

\bigskip

One way to contrast our approach with that of FLM (and FLT) is to think about the constraints that need to be satisfied to decompose a given target payoff $v$ with respect to a given set $V_\mu$. By definition we must find a current action profile $\bm a$ and continuation payoffs $\gamma$.  The achievability condition (that $v$ is the weighted combination of  the utility of the current action profile and the expected continuation values) yields a family of linear equalities.  The incentive compatibility conditions (that players must be deterred from deviating from $\bm a$) yields a family of linear inequalities. In the context of FLM, satisfying all these linear inequalities simultaneously requires a large and rich collection of signals so that many different continuation payoffs can be assigned to different deviations.  Because we have only two signals, we are only able to choose two continuation payoffs but still must satisfy the same family of inequalities -- so our task is much more difficult.  It is this difficulty that leads to the Conditions in Theorem 1.

Note that $\delta_\mu$ is {\em decreasing} in $\mu$.  Since Condition 3 puts an absolute lower bound on $\mu$ and Condition 4 puts an absolute lower bound on $\delta_\mu$ this means that (subject to the regularity constraint) there is a $\mu^*$ such that $V_{\mu^*}$ is the largest self-generating set (of this form) and $\delta_{\mu^*}$ is the smallest discount factor (for which any set of this form can be self-generating).  This may seem puzzling -- increasing the discount factor beyond a point makes no difference -- but remember that we are providing a characterization of self-generating sets and not of PPE payoffs.  However, as we shall see in Theorem 4, for the two-player case, we do obtain a complete characterization of (efficient) PPE payoffs and we demonstrate the same phenomenon.

\section{Perfect Public Equilibrium}

Because every payoff in a self-generating set can be achieved in a PPE, Theorem 1 immediately provides sufficient conditions achieving (some) given target payoffs in perfect public equilibrium.  In fact, we can provide an explicit algorithm for {\em computing} PPE strategies.  A consequence of this algorithm is that (at least when action spaces are finite), the constructed PPE enjoys an interesting and potentially useful robustness property.

\subsection{A Constructing Efficient Perfect Public Equilibria}

Given the various parameters of the environment (game payoffs, information structure, discount factor) and of the problem (lower bound, target vector), the algorithm takes as input in period $t$ the current continuation vector $v(t)$ and computes, for each player $j$, an indicator  $d_j(v(t))$ defined as follows:
$$
d_j(v(t)) = \frac{\lambda_j[v_j(t)-\mu_j]}{\lambda_j[\tilde{v}_j^j - v_j(t)] + \sum_{k\neq j} \lambda_k \, \alpha(j,k) \rho(y_b^j|\bm{\tilde{a}}^j)}
$$
(Note that each player can compute every $d_j$ from the current continuation vector $v(t)$ and the various parameters.)  Having computed $d_j(v(t))$ for each $j$, the algorithm finds the player $i^*$ whose indicator is greatest.   (In case of ties, we arbitrarily choose the player with the largest index.)  The current action profile is $i^*$'s preferred action profile $\bm{\tilde{a}}^{i^*}$.  The algorithm then uses the labeling $Y = \{y^{i^*}_g, y^{i^*}_b\} $ to compute continuation values for each signal in $Y$.

\begin{table}
\renewcommand{\arraystretch}{1.3}
\caption{The algorithm used by each player.} \label{table:EquilibriumStrategy} \centering
\begin{tabular}{l}
\hline
\textbf{Input:} The current continuation payoff $v(t) \in V_{\mu}$ \\
\hline
For each $j$\\
~~~~~~~~Calculate the indicator  $d_j(v(t))$  \\
Find the player $i$ with largest indicator (if a tie, choose largest $i$) \\
~~~~$i = \max_j \left\{ \arg\max_{j\in N} d_j(v(t)) \right\}$ \\
Player $i$ is active; chooses action $\bm{\tilde{a}}^{i}_i$ \\
Players $j \not=i$ are inactive; choose action $\bm{\tilde{a}}^{i}_j$\\
Update $v(t+1)$  as follows: \\
~~~~\textbf{if} $y^t=y_g^{i}$ \textbf{then} \\
~~~~~~~~$v_{i}(t+1)=\tilde{v}_{i}^{i}+(1/\delta) (v_{i}(t)-\tilde{v}_{i}^{i})-(1/\delta - 1) (1/\lambda_{i})\sum_{j\neq {i}} \lambda_j \alpha(i,j) \rho(y_b^{i}|\bm{\tilde{a}}^{i})$ \\
~~~~~~~~$v_j(t+1)=\tilde{v}_j^{i}+(1/\delta) (v_j(t)-\tilde{v}_j^{i})+(1/\delta - 1)\alpha(i,j)\rho(y_b^{i}|\bm{\tilde{a}}^{i})$ \\
~~~~~~~~~~~~ for all $j \not= i$\\
~~~~\textbf{if} $y^t=y_b^{i}$ \textbf{then} \\
~~~~~~~~$v_{i}(t+1)=\tilde{v}_{i}^{i}+(1/\delta) (v_{i}(t)-\tilde{v}_{i}^{i})+(1/\delta - 1) (1/\lambda_{i}) \sum_{j\neq {i}} \lambda_j \alpha(i,j) \rho(y_g^{i}|\bm{\tilde{a}}^{i})$ \\
~~~~~~~~$v_j(t+1)=\tilde{v}_j^{i}+(1/\delta) (v_j(t)-\tilde{v}_j^{i})-(1/\delta - 1) \alpha(i,j) \rho(y_g^{i} | \bm{\tilde{a}}^{i})$\\
~~~~~~~~~~~~ for all $j \not= i$\\
\hline
\end{tabular}
\end{table}

\bigskip

 \begin{theorem}\label{theorem:PPE}  If the conditions in Theorem 1 are satisfied, then  every payoff $v \in V_\mu$ can be achieved in a PPE. For $v \in V_\mu$, a PPE strategy profile that achieves $v$ can be computed by the algorithm in Table~\ref{table:EquilibriumStrategy}
\end{theorem}

\subsection{Robustness}

A consequence of our constructive algorithm is that, for generic values of the parameters of the environment and of the problem and for as many periods as we specify, the strategies we identify are locally constant in these parameters.  To make this precise, we assume for this subsection  that action spaces $A_i$ are finite.   The parameters of the model are the utility mapping $U: {\bm A} \to \R^n$ and the probabilities $\rho(\cdot | \cdot) : Y \times  {\bm A} \to [0,1]$.  Because the probabilities must sum to 1 and we require full support, the parameter space of the model is
$$
 \Omega = \left(R^n \times [0,1] \right)^{\bm A}
 $$
The parameters of the problem are the discount factor $\delta$, the constraint vector $\mu$ and the target profile $v^*$; because the target profile lies in a hyperplane, the parameter space for the particular problem is
$$
\Theta = (0,1) \times \R^n \times \R^{n-1}
$$
Let $\Xi \subset \Omega \times \Theta$ be the subset of parameters that satisfy the Conditions of Theorem 1.  For $\xi \in \Xi$,  the algorithm generates an strategy profile
$$
\sigma_\xi : \hist \to {\bf A}
$$
For $T \geq 0$ we write $\sigma^T_\xi$ for the restriction of $\sigma_\xi$ to the set $\hist^T$  of histories of length at most $T$.

\bigskip

\begin{theorem}  For each  $T \geq 0$ there is a subset  $\Xi_T \subset \Xi$ that is closed and has measure 0 with the property that the mapping $\xi \to \sigma^T_\xi : \Xi \to \hist^T$ is locally constant on the complement of $\Xi_T$.
\end{theorem}

\bigskip

In words: if $\xi, \xi'$ are close together and neither lies in the  proscribed small set of parameters $\Xi_T$, then the strategies $\sigma_\xi, \sigma_{\xi'}$ {\em coincide} for at least the first $T$ periods.

\subsection{Comparison with FLT }

As we have commented in the Introduction, our approach provides a great deal of information about the efficient payoffs that can be achieved in PPE but because the sets $V_\mu$ are required to have a special form, it does not find all of them.  Here we provide a simple example.  We consider a $3 \times 3$ game.  Each player chooses from the actions $\{l,m,h\}$: Player 1 chooses rows, Player 2 chooses columns, Player 3 chooses matrices; see Table \ref{table:3x3}.   (Payoffs indicated by $*$ are irrelevant so long as Assumptions 1,2 are satisfied; we could take $* = 0$ everywhere.)  There are two signals $y_g, y_b$ and the signal structure is
$$
\rho(y_g | \act) = \left\{ \begin{array}{rcl}
                                     2/3 & \mbox{ if } & \act = (h,\ell, \ell) \mbox{ or any permutation } \\
                                     1/2 & \mbox{ if } & \act = (h,m, \ell) \mbox{ or any permutation } \\
                                     1/3 &  &  \mbox{ otherwise } \end{array} \right.
$$
Note that $\tilde{\bm a}^1 = (h,\ell, \ell)$,
$\tilde{\bm a}^2 = (\ell, h, \ell)$, $\tilde{\bm a}^3 = (\ell,\ell, h)$ and that $\tilde{v}^1 = (1, .5, 0)$, $\tilde{v}^2 = (0, 1, .5)$, $\tilde{v}^3 = (.5, 0, 1)$.   Condition 3 implies that no regular $V_\mu$ can be  a self-generating set (because we would have to have $\mu_i > .5$ for each $i$), so our approach does not find any PPE.  However, applying the machinery of FLT shows that there is a discount factor $\hat{\delta} < 1$ for which the payoff vector $(.5,.5,.5)$ -- indeed, any efficient payoff vector close to $(.5,.5,.5)$ -- can be achieved in PPE.\footnote{Calculations available from the authors by request.}  As noted in the Introduction, however, FLT provides no information as to what $\hat{\delta}$ must be nor does it construct PPE strategies.

\begin{table}
\renewcommand{\arraystretch}{1.0}
\caption{Payoff Matrices for the $3 \times 3$ Game; Player 3 Chooses $\ell$, $m$, $h$ (respectively)} \label{table:3x3} \centering
\begin{tabular}{r|c|c|c|}
\multicolumn{1}{r}{}
 &  \multicolumn{1}{c}{}
 & \multicolumn{1}{c}{}
 & \multicolumn{1}{c}{} \\

\multicolumn{1}{r}{}
 &  \multicolumn{1}{c}{$\ell$}
 & \multicolumn{1}{c}{$m$}
 & \multicolumn{1}{c}{$h$} \\
\cline{2-4}
$\ell$ & $(*,*,*)$ & $(*,*,*)$ & $(0, 1, 0.5)$ \\
\cline{2-4}
$m$ & $(*,*,*)$ & $(*,*,*)$ & $(0.1, *, *)$ \\
\cline{2-4}
$h$ & $(1, 0.5, 0)$ & $(*, 0.55, *)$ & $(0.2, 0.6, *)$ \\
\cline{2-4}

\multicolumn{1}{r}{}
 &  \multicolumn{1}{c}{}
 & \multicolumn{1}{c}{}
 & \multicolumn{1}{c}{} \\

 \multicolumn{1}{r}{}
 &  \multicolumn{1}{c}{}
 & \multicolumn{1}{c}{}
 & \multicolumn{1}{c}{} \\

\multicolumn{1}{r}{}
 &  \multicolumn{1}{c}{$\ell$}
 & \multicolumn{1}{c}{$m$}
 & \multicolumn{1}{c}{$h$} \\
\cline{2-4}
$\ell$ & $(*,*,*)$ & $(*,*,*)$ & $(*,*,0.55)$ \\
\cline{2-4}
$m$ & $(*,*,*)$ & $(*,*,*)$ & $(*,*,*)$ \\
\cline{2-4}
$h$ & $(*,*,0.1)$ & $(*,*,*)$ & $(*,*,*)$ \\
\cline{2-4}

\multicolumn{1}{r}{}
 &  \multicolumn{1}{c}{}
 & \multicolumn{1}{c}{}
 & \multicolumn{1}{c}{} \\

 \multicolumn{1}{r}{}
 &  \multicolumn{1}{c}{}
 & \multicolumn{1}{c}{}
 & \multicolumn{1}{c}{} \\

\multicolumn{1}{r}{}
 &  \multicolumn{1}{c}{$\ell$}
 & \multicolumn{1}{c}{$m$}
 & \multicolumn{1}{c}{$h$} \\
\cline{2-4}
$\ell$ & $(0.5,0,1)$ & $(*, 0.1, *)$ & $(*,0.2,0.6)$ \\
\cline{2-4}
$m$ & $(0.55,*,*)$ & $(*,*,*)$ & $(*,*,*)$ \\
\cline{2-4}
$h$ & $(0.6,*,0.2)$ & $(*,*,*)$ & $(*,*,*)$ \\
\cline{2-4}
\end{tabular}
\end{table}

\section{Two Players}

Theorem 1 provides a complete characterization of self-generating sets that have a special form.  If there are only two players then maximal self-generating sets -- the set of all PPE -- have this form and so it is possible to provide a complete characterization of PPE.  We focus on what seems to be the most striking finding:  either there are no efficient PPE outcomes at all (for any discount factor $\delta < 1$) or there is a discount factor $\delta^* < 1$ with the property that any target payoff in $V$ that can be achieved as a PPE for {\em some} $\delta$ can already be achieved for {\em every} $\delta \geq \delta^*$.

\bigskip

\begin{theorem}  Assume $N = 2$ (two players).  Either
\begin{enumerate}[(i)]
\item no target profile in $V$ can be supported in a PPE for any $\delta < 1$ {\em or}
\item there  exist $\overline{\mu}_1,\overline{\mu}_2$  and a discount factor $\delta^* < 1$ such that if $\delta$ is any discount factor with $\delta^* \leq \delta < 1$
then the set of payoff vectors that can be supported in a PPE when the discount factor is $\delta$ is precisely
$$
E = \{v \in V : v_i \geq \overline{\mu}_i \mbox{ for } i = 1, 2 \}
$$
\end{enumerate}
 \end{theorem}

\bigskip

The proof yields explicit (messy) expressions for $\overline{\mu}_1, \overline{\mu}_2$ and $\delta^*$.

\pagebreak

\section{Examples, Redux}\label{sec:example_redux}

In Section 3 we presented three examples to illustrate the model.  We now return to these models to illustrate our analysis and conclusions.

\bigskip

\noindent {\bf Example 1 }  Because there are only two players, Theorem 4 applies.  In this case it is easy to give explicit expressions for
$\overline{\mu}_1, \overline{\mu}_2$ and for the threshold discount factor $\delta^*$.\footnote{Calculations are available from the authors on
request.}  In fact $\overline{\mu}_1 = \overline{\mu}_12 = q/(q-r) b$ and
$$
\delta^* \ =  \ \frac{1}{1+\left(\frac{B-2\frac{q}{q-r}b}{B+2\frac{1-q}{q-r}b}\right)}
$$
so that for every $\delta \geq \delta^*$ the set of efficient payoffs that can be achieved in PPE is exactly
$$
E(\delta) = \{(v_1, v_2) : v_1 + v_2 = B; v_i \geq q/(q-r) b \}
$$
We stress that the set of efficient equilibrium outcomes {\em does not} increase as $\delta \to 1$; as we noted in the Introduction, patience is rewarded but only up to a point.  See Figure \ref{fig:PayoffRegion_PrisonersDilemma_PPE}.

\begin{figure}
\centering
\includegraphics[width =3.0in]{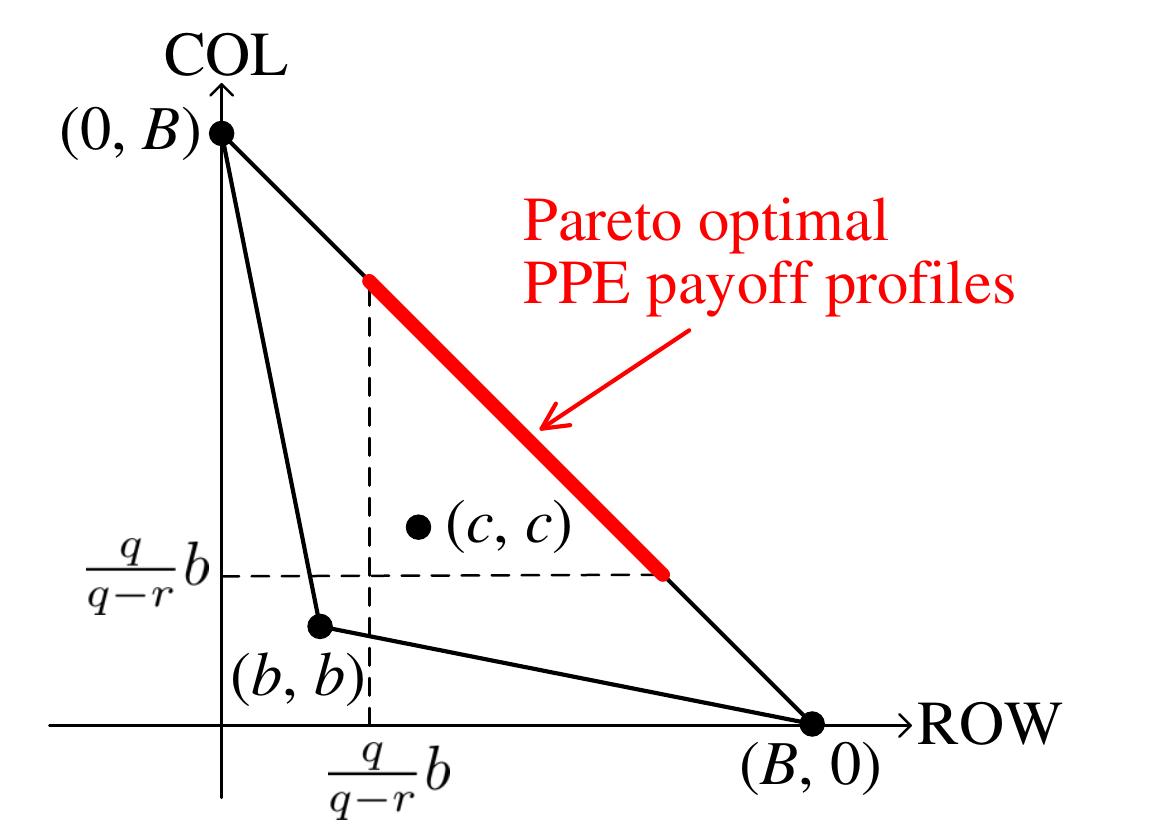}
\caption{Efficient PPE Payoffs for the Modified Prisoners' Dilemma. \label{fig:PayoffRegion_PrisonersDilemma_PPE}}
\end{figure}

\bigskip

\noindent {\bf Example 2 } As we have noted, the choice of a reporting rule has implications for the equilibria of the repeated game.  In this case it is natural to consider two reporting rules $(Y^1, \psi^1)$, $(Y^2, \psi^2)$ where
\begin{itemize}
\item $Y^1 = \{y_b, y_g\}$; $\psi^1(z_k) = y_b$ if $k = 0$,  $\psi^1(z_k) = y_g$ if $k \not= 0$
\item $Y^2  = Z$; $\psi^2(z_k) = z_k$ for all $k = 0, \ldots, n$
\end{itemize}
In the first rule (which is the one discussed in Section 3), the designer announces whether or not there has been a winner; in the second case the designer also announces the identity of the winner.

As we have noted earlier, in the reduced forms of the stage game, the {\em ex ante} expected utilities  are given by
\begin{eqnarray*}
U_i(\act) &=& a_i \left( \eta - \kappa \sum_{j \not= i} a_j \right)^+ R - ca_i \\
\end{eqnarray*}
and, with the first reporting rule,  the signal distribution is
$$
\overline{\rho}(y_* | a) \left\{ \begin{array}{rcl}
                            1 - \sum_i a_i \left(\eta - \kappa \sum_{j \not= i} a_j \right)^+
                                  &\mbox{ if } & * = b \\

                            \sum_i a_i \left(\eta - \kappa \sum_{j \not= i} a_j \right)^+ &\mbox{ if } & *= g
                            \end{array} \right.
$$
With the second reporting rule, the signal distribution is
$$
\widehat{\rho}(y_k | a) = \left\{ \begin{array}{rcl}
                            1 - \sum_i a_i \left(\eta - \kappa \sum_{j \not=i} a_j \right)^+
                                  &\mbox{ if } & k = 0 \\
                            a_k \left(\eta - \kappa \sum_{j \not=k} a_j \right)^+ &\mbox{ if } & k \not= 0
                            \end{array} \right.
$$

To be specific, suppose there are 2 players.  With the first reporting rule, there are two signals; with the second reporting rule there are three signals.  The second reporting rule provides additional information to players and this additional information can be used to support more PPE.  Suppose for instance that a strategy profile $\sigma$ calls for $\tilde{\bf a}^1$ to be played after a particular history.   If all the players follow $\sigma$ then only player $1$ exerts non-zero effort so only two outcomes can occur: either player $1$ wins or no one wins.  If player $2$ deviates by exerting non-zero effort, a third outcome can occur:  $2$ wins.  With either monitoring structure, it is possible for player 1 to detect (statistically) when player 2 has deviated, but with the second monitoring structure it will sometimes be the case that player 1 can be {\em certain} that player 2 has deviated.  This additional information makes it possible to provide additional punishments for deviation and hence to support a larger set of efficient PPE.  As Figure \ref{fig:repeatedcontest} shows, the difference matters:  the second reporting rule always supports a larger set of efficient PPE; indeed, for some values of $\kappa$ (which measures the strength of the interference) only the second reporting rule supports any efficient PPE at all.

\begin{figure}
\centering
\includegraphics[width =5.0in]{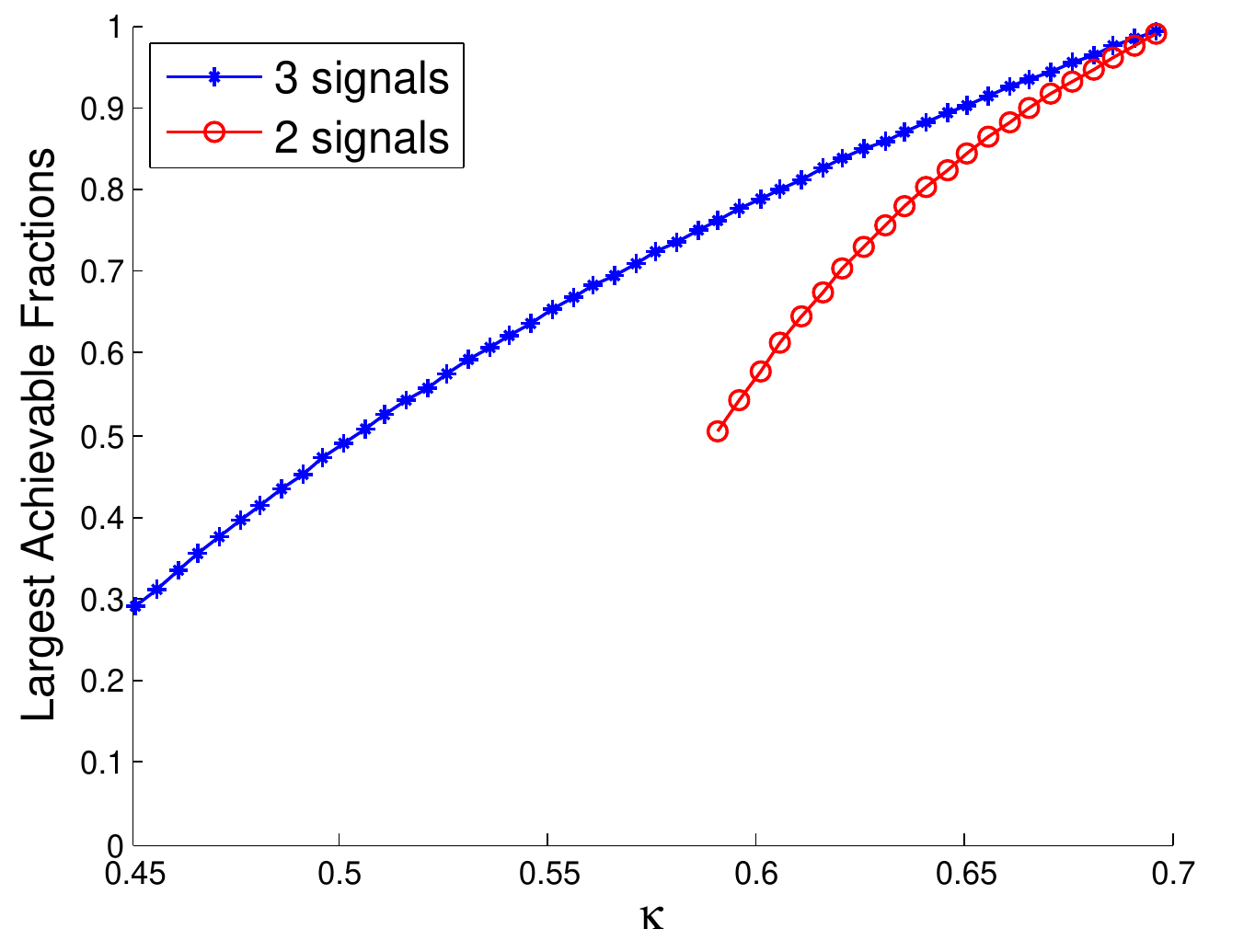}
\caption{Largest Achievable Fraction $1-\mu_i/\tilde{v}_i^i$ as a Function of $\kappa$. \label{fig:repeatedcontest}}
\end{figure}

More generally, consider a context with $n \geq 3$ players.  If the identity of the winner is not announced, all deviations must be ``punished'' in the same way, but if the identity of the winner is announced, ``punishments'' can be tailored to the deviator, and hence can be more severe. If punishments can be more severe it may be possible to sustain a wider range of PPE.  Which reporting rule -- hence which monitoring structure -- should be chosen by the designer will depend on the tradeoff the designer makes between preserving privacy and sustaining a wider range of PPE.

\bigskip

\noindent {\bf Example 3 }  As we have suggested, the designer must choose from some (possible) measurement technologies. This choice involves a tension: more accurate measurement technologies will typically be more costly to employ.  Hence the designer must trade-off the accuracy of the measurement technology against the cost of employing it.  The designer must also choose a reporting rule, in this case a threshold $d_0$. This choice also involves a tension, but of a different kind.  Given the distribution of shocks and a choice of measurement technology, the choice of threshold affects the distribution of signals.  How the designer chooses the distribution of signals depends on what the designer wishes to accomplish.  For instance, given a fixed discount factor, the designer may wish to choose the threshold to maximize the range of long-run resource allocations that can be supported as PPE for the given discount factor.  Alternatively, the designer may wish to minimize the discount factor for which {\em some} long-run resource allocation can be supported as a PPE.

To give some idea of the effect of these tradeoffs, we present numerical results for a special case of the Resource Sharing Game with 3 players, capacity $\chi = 1$ and $\bar{\varepsilon}=0.3$.  Because the game is symmetric it seems natural to consider symmetric sets of payoffs; so we consider sets of the form
$$
V(\eta) = \{v \in V : v_i \geq \eta\tilde{v} \mbox{ for each } i \}
$$
where $\tilde{v}$ is the utility of each player's most preferred action and $\eta \in [0,1]$.     Note that $1-\eta$ represents the fraction of the entire efficient set $V$ that is occupied by $V(\eta)$.  A natural {\em desideratum} for the designer is to choose the threshold $d_0$ so that the fraction $1-\eta$ is as large as possible; this maximizes opportunities for sharing.  (As we have shown in Theorem 1,  making $\eta$ smaller also makes the required discount factor smaller, so the designer can simultaneously create more sharing opportunities for less patient players.)       Figures \ref{fig:LowerBoundPayoff} and \ref{fig:LowerBoundDiscountFactor} display (from simulations) the relationship between the threshold $d_0$ and the smallest $\eta$ and smallest $\delta$ for different values of the exponent $p$.

\begin{figure}
\centering
\includegraphics[width =5.0in]{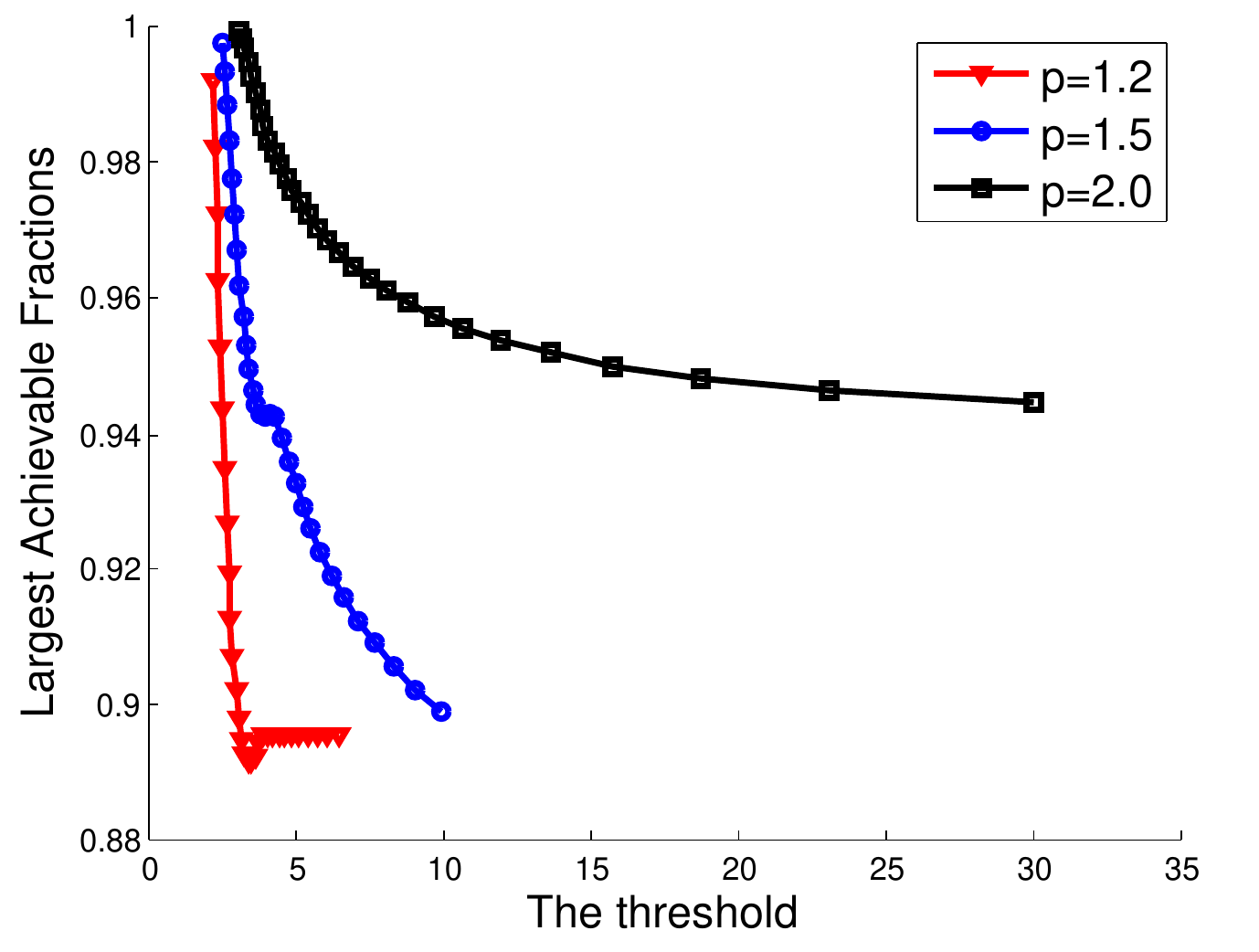}
\caption{Largest Achievable Fraction $1- \eta$ as a Function of Threshold $d_0$. \label{fig:LowerBoundPayoff}}
\end{figure}

\begin{figure}
\centering
\includegraphics[width =5.0in]{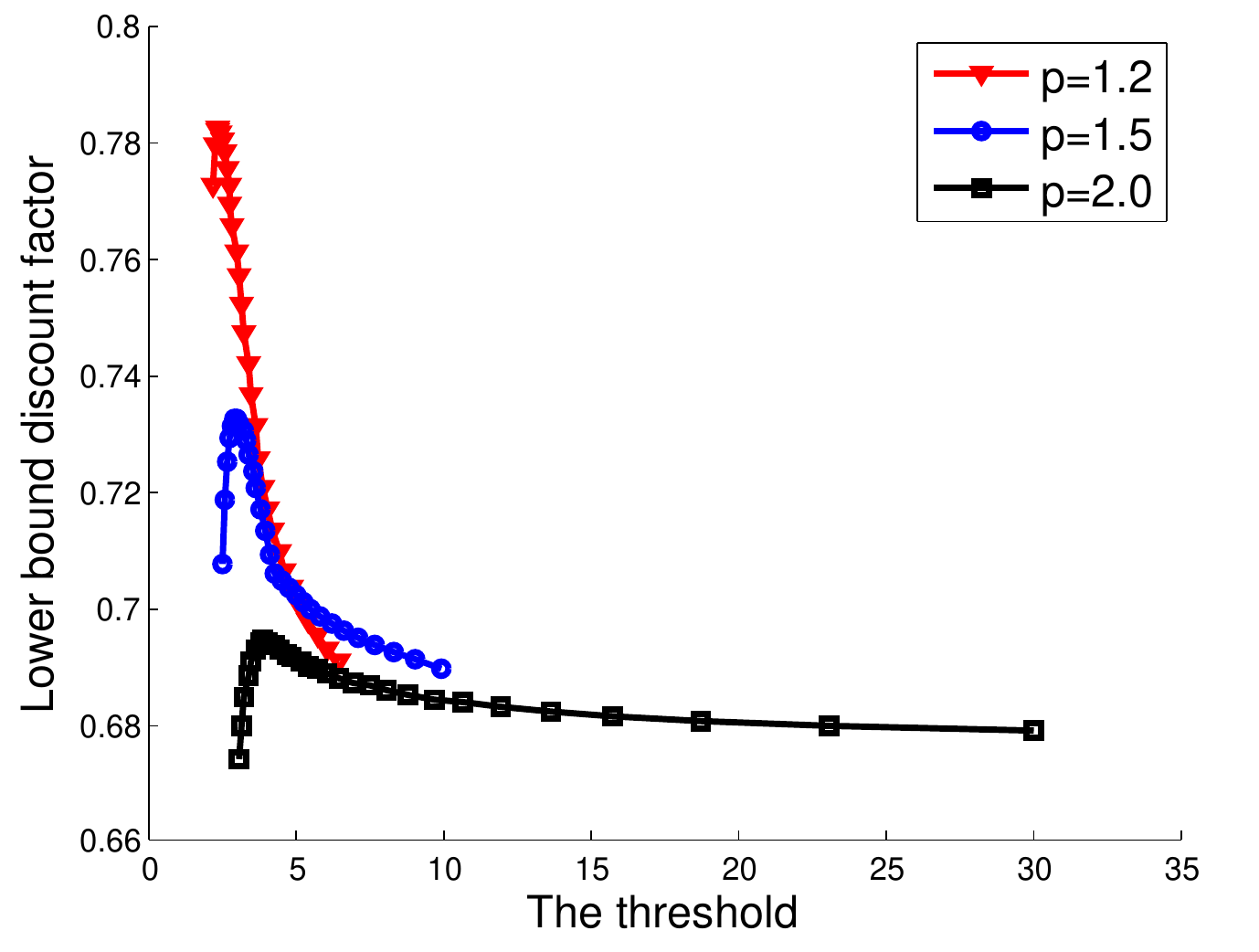}
\caption{Smallest Achievable Discount Factor $\delta$ as a Function of Threshold $d_0$ \label{fig:LowerBoundDiscountFactor}}
\end{figure}

\section{Conclusion}  This paper diverges from much of the familiar literature on repeated games with imperfect public monitoring in two directions.  In analyzing the reduced form, we make different assumptions on the signal structure and obtain stronger conclusions about efficient PPE (bounds on the discount factor, explicitly constructive strategies).  However, we also construct an elaborated form in which the information structure can be viewed as arising from the behavior of a strategic designer. Clearly there is much more to be done.  Perhaps most obviously, it is clearly important to understand the extent to which the assumptions on the signal structure of the reduced form and on the geometry of the candidate self-generating sets $V_\mu$ can be relaxed.  However, we think the elaborated form is of even more potential interest, especially for applications.  As we have discussed in the Examples, the designer must decide what to observe and what to communicate to the players, and these choices will typically involve a trade-off between the cost of more accurate observation and communication on the one hand and the benefits of better information on the other hand.  It seems natural to suppose that the costs and benefits -- and hence the trade-offs -- may be very different across environments.  This seems a subject worthy of much study.


%
%
%

\bibliographystyle{elsarticle-harv}
\bibliography{RepeatedGames}

\pagebreak

\section*{Appendix}

The proof of Proposition 1 is immediate and omitted.

\bigskip

\noindent {\bf Proof of Proposition 2 } Fix an active player $i$ and an inactive player $j$.  Set
\begin{eqnarray*}
A(i,j) &=& \big\{ a_j \in A_j : u_j(a_j, \besti_{-j}) > u_j(\besti) \big\} \\
B(i,j) &=& \big\{ a_j \in A_j : u_j(a_j, \besti_{-j}) <  u_j(\besti) , \rho(y^i_b | a_j, \besti_{-j}) < \rho(y^i_b |  \besti) \big\}
\end{eqnarray*}
If either of $A(i,j)$ or $B(i,j)$ is empty then $\alpha(i,j) \leq \beta(i,j)$ by default, so assume in what follows that neither of $A(i,j)$, $B(i,j)$ is empty.

Fix a discount factor $\delta \in (0,1)$ and let $\sigma$ be PPE that achieves an efficient payoff.  Assume that $i$ is active following some history:
$\sigma(h) = \besti$ for some $h$.   Because $\sigma$ achieves an efficient payoff, we can decompose the payoff $v$ following $h$ as the weighted sum of the current payoff from $\besti$ and the continuation payoff assuming that players follow $\sigma$; because $\sigma$ is a PPE, the incentive compatibility condition for all players $j$ must obtain.  Hence for all $a_j \in A_j$ we have
\begin{eqnarray}\label{eqn:Decomposition}
v_j &=& (1-\delta) u_j(\besti) + \delta  \sum_{y\in Y} \rho(y|\besti)  \gamma_j(y) \nonumber \\
    &\geq& (1-\delta)u_j(a_j,\besti_{-j}) + \delta  \sum_{y\in Y} \rho(y|a_j,\besti_{-j})  \gamma_j(y),
\end{eqnarray}
Substituting probabilities for the good and bad signals yields
\begin{eqnarray*}
v_j &=& (1-\delta) u_j(\besti) + \delta  \big[\rho(y_g^i|\besti)  \gamma_j(y_g^i)+\rho(y_b^i|\besti)  \gamma_j(y_b^i)\big] \\
    &\geq& (1-\delta) u_j(a_j,\besti_{-j}) + \delta  \big[ \rho(y_g^i|a_j,\besti_{-j})  \gamma_j(y_g^i)+\rho(y_b^i|a_j,\besti_{-j})  \gamma_j(y_b^i)\big]\nonumber
\end{eqnarray*}
Rearranging yields
\begin{eqnarray*}
\big[\rho(y_b^i|a_j,\besti_{-j})-\rho(y_b^i|\besti)\big]  \big[\gamma_j(y_g^i) - \gamma_j(y_b^i)\big] \big[\frac{\delta}{1- \delta}\big]
   \geq   \big[ u_j(a_j,\besti_{-j})- u_j(\besti) \big]
\end{eqnarray*}
Now suppose $j \not= i$ is an inactive player.
If $a_j \in A(i,j)$ then $\rho(y^i_b | a_j, \besti_{-j}) - \rho(y^i_b |  \besti) > 0$ (by Assumption 4) so
\begin{eqnarray}\label{eqn:A(i,j)}
  \big[\gamma_j(y_g^i) - \gamma_j(y_b^i)\big] \big[\frac{\delta}{1- \delta}\big]
   \geq   \frac{ u_j(a_j,\besti_{-j})- u_j(\besti) }{\rho(y_b^i|a_j,\besti_{-j})-\rho(y_b^i|\besti)}
\end{eqnarray}
If $a_j \in B(i,j)$ then $\rho(y^i_b | a_j, \besti_{-j}) - \rho(y^i_b |  \besti) < 0$ (by definition) so
\begin{eqnarray}\label{eqn:B(i,j)}
  \big[\gamma_j(y_g^i) - \gamma_j(y_b^i)\big] \big[\frac{\delta}{1- \delta}\big]
   \leq   \frac{ u_j(a_j,\besti_{-j})- u_j(\besti) }{\rho(y_b^i|a_j,\besti_{-j})-\rho(y_b^i|\besti)}
\end{eqnarray}
Taking the sup over $a_j \in A(i,j)$ in (\ref{eqn:A(i,j)}) and the inf over $a_j \in B(i,j)$ in (\ref{eqn:B(i,j)}) yields $\alpha(i,j) \leq \beta(i,j)$ as desired.
$\Box$

\bigskip

\noindent {\bf Proof of Proposition 3 }  As above, we assume $i$ is active following the history $h$ and that $v$ is the payoff following $h$.  Fix $a_i  \in A_i$.  By definition, $u_i(\besti) > u_i(a_i, \besti_{-i})$.  With respect to probabilities,  there are two possibilities.  If $\rho(y^i_b | a_i, \besti_{-i}) \leq \rho(y^i_b | \besti)$ then we immediately have
$$
\tilde{v}_i^i-u_i(a_i,\besti_{-i}) \geq \frac{1}{\lambda_i} \sum_{j\neq i} \lambda_j
\alpha(i,j)[\rho(y_b^i|a_i,\besti_{-i})-\rho(y_b^i|\besti)]
$$
because the left-hand side is positive and the right-hand side is non-negative.  If  $\rho(y^i_b | a_i, \besti_{-i}) > \rho(y^i_b | \besti)$ we proceed as follows.

 We begin with  \eqref{eqn:Decomposition} but now we apply it to the active user $i$, so that for all $a_i\in A_i$ we have
\begin{eqnarray*}
v_i &=& (1-\delta) u_i(\besti) + \delta  \big[ \rho(y_g^i|\besti)  \gamma_i(y_g^i)+\rho(y_b^i|\besti)  \gamma_i(y_b^i) \big] \\
    &\geq& (1-\delta) u_i(a_i,\besti_{-i}) + \delta  \big[ (\rho(y_g^i|a_i,\besti_{-i})  \gamma_i(y_g^i)+\rho(y_b^i|a_i,\besti_{-i})  \gamma_i(y_b^i) \big] \nonumber
\end{eqnarray*}
Rearranging yields
$$
  \gamma_i(y_g^i) - \gamma_i(y_b^i)\
    \geq \left[ \frac{1-\delta}{\delta} \right] \left[ \frac{u_i(a_i,\besti_{-i}) - u_i(\besti) }{ \rho(y_b^i|a_i,\besti_{-i})-\rho(y_b^i|\besti)} \right]
$$
Because continuation payoffs are in $V$, which lies in the hyperplane $H$, the continuation payoffs for the active user can be expressed in terms of the continuation payoffs for the inactive users as
$$
\gamma_i(y)=\frac{1}{\lambda_i}\left[1-\sum_{j\neq i} \lambda_j \gamma_j(y)\right]
$$
 Hence
$$
\gamma_i(y_g^i) - \gamma_i(y_b^i) = -\frac{1}{\lambda_i} \sum_{j\neq i} \lambda_j [\gamma_j(y_g^i) - \gamma_j(y_b^i)]
$$
Applying the incentive compatibility constraints for the  inactive users implies that for each $a_j \in A(i,j)$ we have
$$
\gamma_j(y_g^i) - \gamma_j(y_b^i) \geq \left[ \frac{1-\delta}{\delta} \right]
\left[\frac{u_j(a_j,\besti_{-j})-u_j(\besti)}{\rho(y_b^i|a_j,\besti_{-j})-\rho(y_b^i|\besti)} \right]
$$
In particular
$$
\gamma_j(y_g^i) - \gamma_j(y_b^i) \geq \left[ \frac{1-\delta}{\delta} \right] \alpha(i,j)
$$
and hence
$$
\gamma_i(y_g^i) - \gamma_i(y_b^i) \leq -\frac{1}{\lambda_i} \left[\frac{1-\delta}{\delta}\right] \left[ \sum_{j\neq i} \lambda_j  \alpha(i,j)\right] \leq 0
$$
Putting these all together, canceling the factor $[1-\delta]/\delta$ and remembering that we are in the case  $\rho(y^i_b | a_i, \besti_{-i}) > \rho(y^i_b | \besti)$ yields
$$
\tilde{v}_i^i-u_i(a_i,\besti_{-i}) \geq \frac{1}{\lambda_i} \sum_{j\neq i} \lambda_j
\alpha(i,j)[\rho(y_b^i|a_i,\besti_{-i})-\rho(y_b^i|\besti)]
$$
which is the desired result.  $\Box$

\bigskip

\noindent {\bf Proof of Theorem 1 }  Assume that $V_\mu$ is regular and not an extreme point, and is a self-generating set; we verify Conditions 1-4 in turn.  Because $V_\mu$ is self-generating and not an extreme point, it cannot be a singleton and hence must contain an interior point of $V$.  In order for such a point to be achieved in a PPE, every player must be active following some history, so Propositions 2 and 3 yield Conditions 1 and 2.

By assumption, for each $i \in N$ there is a payoff profile $\hat{v}^i \in V_\mu$ with the property that $\hat{v}^i_j = \mu_j$ for each $j \not= i$.  Necessarily, $\hat{v}^i$ is the unique such point and $\hat{v}^i = \argmax \{v_i : v \in V_\mu \}$.  Because $V$ lies in the hyperplane $H$ we have
$$
\hat{v}_j^i = \left\{\begin{array}{ccl} \mu_j & \mbox{ if } & j \not= i \\ \frac{1}{\lambda_i}\left(1-\sum_{k\neq i} \lambda_k \mu_k\right)& \mbox{ if } & j=i
\end{array}\right. \nonumber
$$

Because $V_\mu$ is self-generating, we can decompose $\hat{v}^i$:
\begin{equation}\label{eqn:decompose}
\hat{v}^i = (1-\delta)u(\tilde{\act}^k) + \delta \sum_y \rho(y| \tilde{\act}^k) \gamma(y)
\end{equation}
for some $\tilde{\act}^k$.  If $k \not= i$ then (because $V_\mu \not= \{\tilde{v}^k \}$) we must have $\mu_k < \tilde{v}^k_k$ which implies that $\gamma_k(y) < \mu_k$ for some $y$; since continuation payoffs must lie in $V_\mu$ this is a contradiction.  Hence in the decomposition (\ref{eqn:decompose}) we must have $\tilde{\act}^k = \tilde{\act}^i$.

It is convenient to first establish the following inequality on $\mu_j$ on the way to establishing the  bounds in Condition 3.
$$
\mu_j > \max_{i\neq j} \tilde{v}_j^i \mbox{ for all } j\in N
$$
To see this, suppose to the contrary that there exists a $i,j$ such that $\mu_j \leq \tilde{v}_j^i$. Consider $i$'s preferred payoff profile $\hat{v}^i$ in $V_{\mu}$. Because decomposing $\hat{v}^i$ requires that we use $\besti$, it follows that
$$
\mu_j = (1-\delta)\cdot \tilde{v}_j^i + \delta\cdot \sum_{y} \rho(y|\besti) \gamma_j(y)
$$
If $\mu_j < \tilde{v}_j^i$ then  $\sum_{y\in Y} \rho(y|\besti)\gamma_i(y) < \mu_j$ and so  $\gamma_j(y)<\mu_j$ for some $y$. This contradicts that fact that $\gamma(y)\in V_{\mu}$.  If $\mu_j = \tilde{v}_j^i$, we must have $\sum_{y} \rho(y|\besti) \gamma_j(y) = \mu_j$. Since $\gamma_j(y)\geq \mu_j$ for all $y$, we must
have $\gamma_j(y_g^i)=\gamma_j(y_b^i)=\mu_j$. By assumption,  player $j$ has a currently profitable deviation $a_j$ so that
$u_j(a_j,\besti_{-j})>u_j(\besti)$, which implies that the continuation payoff $\gamma_j(y_g^i)=\gamma_j(y_b^i)=\mu_j$ cannot satisfy the incentive compatibility  constraints. Hence, we must have $\mu_j > \tilde{v}_j^i$ as asserted.

With all this in hand we derive Condition 3.  To do this, we suppose $i$ is active and examine the decomposition of the inactive player $j$'s payoff in greater detail.  Because $\mu_j > \tilde{v}^i_j$ and $v_j \geq \mu_j$ for every $v \in V_\mu$ we certainly have $v_j > \tilde{v}^i_j$.  We can write $j$'s incentive compatibility condition as
\begin{eqnarray}\label{eqn:IC_inactive}
v_j &=& (1-\delta)\cdot \tilde{v}_j^i + \delta \cdot \sum_{y\in Y} \rho(y|\besti) \cdot \gamma_j(y) \\
    &\geq& (1-\delta)\cdot u_j(a_j,\besti_{-j}) + \delta \cdot \sum_{y\in Y} \rho(y|a_j,\besti_{-j}) \cdot \gamma_j(y). \nonumber
\end{eqnarray}
From the equality constraint in \eqref{eqn:IC_inactive}, we can solve for  the discount factor $\delta$ as
$$
\delta = \frac{v_j-\tilde{v}_j^i}{\sum_{y\in Y} \gamma_j(y)\rho(y|\besti) - \tilde{v}_j^i}
$$
(Note that the denominator can never be zero and the above equation is well defined, because $v_j>\tilde{v}_j^i$ implies that $\sum_{y\in Y}
\gamma_j(y)\rho(y|\besti) > \tilde{v}_j^i$.) We can then eliminate the discount factor $\delta$ in the inequality
of \eqref{eqn:IC_inactive}. Since $v_j>\tilde{v}_j^i$, we can obtain equivalent inequalities, depending on whether $a_j$ is a profitable or unprofitable current deviation):
\begin{itemize}
\item If $ u_j(a_j,\bm{\tilde{a}}_{-j}^i)>\tilde{v}_j^i $ then
\begin{eqnarray}
 v_j &\leq& \sum_{y\in Y}
\gamma_j(y)\Big[\left(1-\frac{v_j-\tilde{v}_j^i}{u_j(a_j,\besti_{-j})-\tilde{v}_j^i}\right)\rho(y|\besti) \nonumber \\
&&  \hspace{.8in} + \ \frac{v_j-\tilde{v}_j^i}{u_j(a_j,\besti_{-j})-\tilde{v}_j^i}\rho(y|a_j,\besti_{-j}) \Big] \label{eqn:ICactive-a}
\end{eqnarray}
\item If $ u_j(a_j,\bm{\tilde{a}}_{-j}^i)< \tilde{v}_j^i $ then
\begin{eqnarray}
 v_j &\geq& \sum_{y\in Y}
\gamma_j(y)\Big[\left(1-\frac{v_j-\tilde{v}_j^i}{u_j(a_j,\besti_{-j})-\tilde{v}_j^i}\right)\rho(y|\besti) \nonumber \\
&& \hspace{.8in} + \ \frac{v_j-\tilde{v}_j^i}{u_j(a_j,\besti_{-j})-\tilde{v}_j^i}\rho(y|a_j,\besti_{-j}) \Big] \label{eqn:ICactive-b}
\end{eqnarray}
\end{itemize}

For notational convenience,  write the coefficient of $\gamma_j(y_g^i)$ in the above inequalities as
\begin{eqnarray*}
c_{ij}(a_j,\bm{\tilde{a}}_{-j}^i)
&\triangleq&  \left(1-\frac{v_j-\tilde{v}_j^i}{u_j(a_j,\besti_{-j})-\tilde{v}_j^i}\right)\rho(y_g^i|\besti) \nonumber  \\
&& \hspace{.5in} +\ \left(\frac{v_j-\tilde{v}_j^i}{u_j(a_j,\besti_{-j})-\tilde{v}_j^i}\right)\rho(y_g^i|a_j,\besti_{-j}) \nonumber \\
&=& \rho(y_g^i|\besti)+(v_j-\tilde{v}_j^i) \left(\frac{\rho(y_g^i|a_j,\besti_{-j})-\rho(y_g^i|\besti)}{u_j(a_j,\besti_{-j})-\tilde{v}_j^i} \right) \nonumber \\
&=& \rho(y_g^i|\besti)-(v_j-\tilde{v}_j^i) \left( \frac{\rho(y_b^i|a_j,\besti_{-j})-\rho(y_b^i|\besti)}{u_j(a_j,\besti_{-j})-\tilde{v}_j^i} \right)
\end{eqnarray*}
According to (\ref{eqn:ICactive-a}), if $u_j(a_j,\bm{\tilde{a}}_{-j}^i)>\tilde{v}_j^i$ then
\begin{eqnarray}\label{eqn:IC_active_cij}
c_{ij}(a_j,\besti_{-j}) \cdot \gamma_j(y_g^i) +\big[1-c_{ij}(a_j,\besti_{-j}) \big]  \gamma_j(y_b^i) \leq v_j
\end{eqnarray}
Since $\gamma_j(y_g^i)>\gamma_j(y_b^i)$, this is true if and only if
\begin{eqnarray}\label{eqn:IC_active_cij+}
\kappa_{ij}^+ \cdot \gamma_j(y_g^i) + (1-\kappa_{ij}^+) \cdot \gamma_j(y_b^i) \leq v_j,
\end{eqnarray}
where $\kappa_{ij}^+  \triangleq \sup \{ c_{ij}(a_j,\bm{\tilde{a}}_{-j}^i) : a_j\in A_j: u_j(a_j,\bm{\tilde{a}}_{-j}^i)>\tilde{v}_j^i \}$.  (Fulfilling the inequalities \eqref{eqn:IC_active_cij} for all $a_j$ such that $u_j(a_j,\besti_{-j})>u_j(\besti)$ is equivalent
to fulfilling the single inequality \eqref{eqn:IC_active_cij+}.  If \eqref{eqn:IC_active_cij+} is satisfied, then the inequalities \eqref{eqn:IC_active_cij} are  satisfied for all $a_j$ such that $u_j(a_j,\besti_{-j})>u_j(\besti)$ because
$\gamma_j(y_g^i)>\gamma_j(y_b^i)$ and $\kappa_{ij}^+\geq c_{ij}(a_j,\besti_{-j})$ for all $a_j$ such that $u_j(a_j,\besti_{-j})>u_j(\besti)$. Conversely,  if the inequalities \eqref{eqn:IC_active_cij} are satisfied for all $a_j$ such that $u_j(a_j,\besti_{-j})>u_j(\besti) $ and  \eqref{eqn:IC_active_cij+} were violated, so that
$\kappa_{ij}^+ \cdot \gamma_j(y_g^i) + (1-\kappa_{ij}^+) \cdot \gamma_j(y_b^i) > v_j$, then we can find a $\kappa_{ij}^\prime<\kappa_{ij}^+$ such
that $\kappa_{ij}^\prime \cdot \gamma_j(y_g^i) + (1-\kappa_{ij}^\prime) \cdot \gamma_j(y_b^i) > v_j$. Based on the definition of the supremum, there
exists at least a $a_j^\prime$ such that $u_j(a_j^\prime,\besti_{-j})>u_j(\besti)$ and $c_{ij}(a_j^\prime,\besti_{-j})>c_{ij}^\prime$, which means
that $c_{ij}(a_j^\prime,\besti_{-j}) \cdot \gamma_j(y_g^i) + (1-c_{ij}(a_j^\prime,\besti_{-j})) \cdot \gamma_j(y_b^i) > v_j$. This contradicts the
fact that the inequalities \eqref{eqn:IC_active_cij+} are fulfilled for all $a_j$ such that $u_j(a_j,\besti_{-j})>u_j(\besti)$.)

Similarly, according to (\ref{eqn:ICactive-b}), for all $a_j$ such that $u_j(a_j,\bm{\tilde{a}}_{-j}^i)<\tilde{v}_j^i$, we must have
$$
c_{ij}(a_j,\bm{\tilde{a}}_{-j}^i)  \gamma_j(y_g^i) + [1-c_{ij}(a_j,\bm{\tilde{a}}_{-j}^i)]  \gamma_j(y_b^i) \geq v_j.
$$
Since $\gamma_j(y_g^i)>\gamma_j(y_b^i)$, the above requirement is fulfilled if and only if
$$
\kappa_{ij}^- \cdot \gamma_j(y_g^i) + (1-\kappa_{ij}^-) \cdot \gamma_j(y_b^i) \geq v_j,
$$
where $\kappa_{ij}^- \triangleq \inf \big\{c_{ij}(a_j,\bm{\tilde{a}}_{-j}^i) : a_j\in A_j,  u_j(a_j,\bm{\tilde{a}}_{-j}^i)<\tilde{v}_j^i \big\}$.
Hence, the decomposition \eqref{eqn:IC_inactive} for user $j\neq i$ can be simplified as:
\begin{eqnarray}\label{eqn:IC_inactive_simplified}
\rho(y_g^i|\besti) \cdot \gamma_j(y_g^i) + [1-\rho(y_g^i|\besti)]  \gamma_j(y_b^i) &=&  \tilde{v}_j^i  + \frac{v_j-\tilde{v}_j^i}{\delta} \nonumber \\
\kappa_{ij}^+ \, \gamma_j(y_g^i) + (1-\kappa_{ij}^+) \cdot \gamma_j(y_b^i) &\leq& v_j \nonumber \\
\kappa_{ij}^-  \, \gamma_j(y_g^i) + (1-\kappa_{ij}^-) \cdot \gamma_j(y_b^i) &\geq& v_j
\end{eqnarray}

Keep in mind that the various continuation values $\gamma$ and the expressions $\kappa_{ij}^+, \kappa_{ij}^-$ depend on $v_j$; where necessary we write the dependence explicitly.  Note that there could be many $\gamma_j(y_g^i)$ and $\gamma_j(y_b^i)$ that satisfy \eqref{eqn:IC_inactive_simplified}. For a given discount factor $\delta$, we
call all the continuation payoffs that satisfy \eqref{eqn:IC_inactive_simplified} {\em feasible} -- but whether particular continuation values lie in $V_\mu$ depends on the discount factor.

We assert that $\kappa_{ij}^+(\mu_j)\leq0$ for all $i\in N$ and for all $j\neq i$.  To see this,  we look again at player $i$'s preferred payoff profile
$\hat{v}^i$ in $V_{\mu}$, which is necessarily decomposed by $\besti$.   We look at the following constraint for player $j\neq i$ in
\eqref{eqn:IC_inactive_simplified}:
$$
\kappa_{ij}^+ \  \gamma_j(y_g^i) + (1-\kappa_{ij}^+)  \ \gamma_j(y_b^i) \leq \mu_j.
$$
Suppose that $\kappa_{ij}^+(\mu_j)>0$. Since player $j$ has a currently profitable deviation from $\besti$, we must set $\gamma_j(y_g^i)>
\gamma_j(y_b^i)$. Then to satisfy the above inequality, we must have $\gamma_j(y_b^i)<\mu_j$. In other words, when $\kappa_{ij}^+(\mu_j)>0$, all the
feasible continuation payoffs of player $j$ must be outside $V_{\mu}$. This contradicts the fact that $V_{\mu}$ is self-generating so the assertion follows.

The definition of  $\kappa_{ij}^+(\mu_j)$ and the fact that $\kappa_{ij}^+(\mu_j) \leq 0$ entail that
\begin{eqnarray*}
\kappa_{ij}^+(\mu_j) &=& \rho(y_g^i|\besti) - (\mu_j-\tilde{v}_j^i) \inf_{a_j \in A(i,j)} \left[ \frac{\rho(y_b^i|a_j,\besti_{-j}) - \rho(y_b^i|\besti)}{u_j(a_j,\besti_{-j})-\tilde{v}_j^i}  \right] \\
&=& \rho(y_g^i|\besti)-(\mu_j-\tilde{v}_j^i) \left[\frac{1}{\sup_{a_j\in A(i,j)}
\left(\frac{u_j(a_j,\besti_{-j})-\tilde{v}_j^i}{\rho(y_b^i|a_j,\besti_{-j})-\rho(y_b^i|\besti)} \right)} \right] \nonumber \\
&=& \rho(y_g^i|\besti)-(\mu_j-\tilde{v}_j^i) \left[ \frac{1}{\alpha(i,j)} \right] \nonumber \\
&\leq& 0 \nonumber
\end{eqnarray*}
This provides a lower bound on $\mu_j$:
$$
\mu_j \geq \tilde{v}_j^i+\alpha(i,j) \rho(y_g^i|\besti) = \tilde{v}_j^i+\alpha(i,j)[1-\rho(y_b^i|\besti)]
$$
This bound must hold for every $i\in N$ and every $ j\not=i$. Hence, we have
$$
\mu_j \geq \max_{i\neq j} \left( \tilde{v}_j^i + \alpha(i,j) [1-\rho(y_b^i|\besti)] \right)
$$
which is Condition~3.

Now we derive Condition 4 (the necessary condition on the discount factor). The minimum discount factor $\underline{\delta}_{\mu}$ required
for $V_{\mu}$ to be a self-generating set solves the optimization problem
$$
\underline{\delta}_{\mu}
 = \max_{v\in V_{\mu}} \delta   ~~~~\mathrm{subject~to}~v \in \mathscr{B}(V_{\mu};\delta)
$$
where $\mathscr{B}(V_{\mu};\delta)$ is the set of payoff profiles that can be decomposed on $V_{\mu}$ under discount factor $\delta$. Since
$\mathscr{B}(V_{\mu};\delta)=\cup_{i\in\mathcal{N}} \mathscr{B} (V_{\mu};\delta,\besti)$, the above optimization problem can be reformulated as
\begin{eqnarray}\label{eqn:delta_optimization}
\underline{\delta}_{\mu} =\max_{v\in V_{\mu}} \min_{i\in N} \delta  ~~~~~\mathrm{subject~to}~v \in \mathscr{B}(V_{\mu};\delta,\besti).
\end{eqnarray}
To solve the optimization problem \eqref{eqn:delta_optimization}, we explicitly express the constraint $v\in\mathscr{B}(V_{\mu};\delta,\besti)$ using
the results derived above.

\begin{figure} \centering
\includegraphics[width =5.0in]{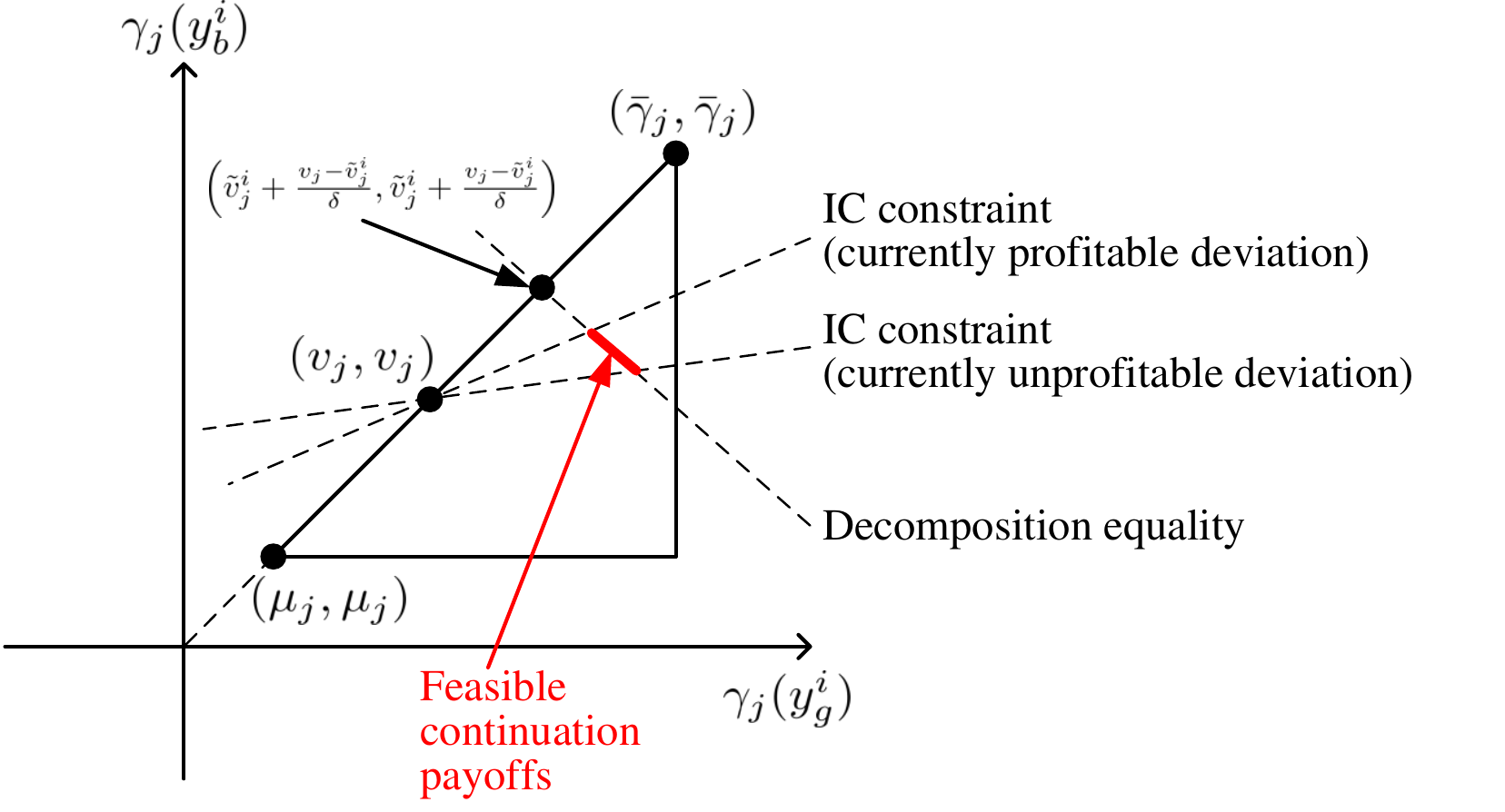}
\caption{Illustrations of the feasible continuation payoffs when $\kappa_{ij}^+\leq0$. $\bar{\gamma}_j=\frac{1}{\lambda_j}\left(1-\sum_{k\neq j}
\lambda_k \mu_k\right)$.} \label{fig:Illustration_ContinuationPayoffs}
\end{figure}

Some intuition may be useful.  Suppose that $i$ is active and $j$ is an inactive player. Recall that player $j$'s feasible $\gamma_j(y_g^i)$ and $\gamma_j(y_b^i)$ must satisfy  (\ref{eqn:IC_inactive_simplified}).
There are many $\gamma_j(y_g^i)$ and $\gamma_j(y_b^i)$ that satisfy  (\ref{eqn:IC_inactive_simplified}). In
Fig.~\ref{fig:Illustration_ContinuationPayoffs}, we show the feasible continuation payoffs that satisfy  (\ref{eqn:IC_inactive_simplified}) when
$\kappa_{ij}^+(v_j)\leq0$. We can see that all the continuation payoffs on the heavy line segment are feasible. The line segment is on the line that
represents the decomposition equality $\rho(y_g^i|\besti) \cdot \gamma_j(y_g^i) + (1-\rho(y_g^i|\besti)) \cdot \gamma_j(y_b^i) = \tilde{v}_j^i +
\frac{v_j-\tilde{v}_j^i}{\delta}$, and is bounded by the IC constraint on currently profitable deviations
$\kappa_{ij}^+ \cdot \gamma_j(y_g^i) + (1-\kappa_{ij}^+) \cdot \gamma_j(y_b^i) \leq v_j$ and the IC constraint on currently unprofitable deviations
$\kappa_{ij}^- \cdot \gamma_j(y_g^i) + (1-\kappa_{ij}^-) \cdot \gamma_j(y_b^i) \geq v_j$.  Among all the feasible continuation payoffs, denoted $\gamma^\prime(y)$, we choose the one, denoted $\gamma^*(y)$, such that  for all $j\not=i$,
$\gamma_j^*(y_g^i)$ and $\gamma_j^*(y_b^i)$ make the IC constraint on currently profitable deviations in  (\ref{eqn:IC_inactive_simplified})
binding. This is because under the same discount factor $\delta$, if there is any feasible continuation payoff $\gamma^\prime(y)$ in the
self-generating set, the one that makes the IC constraint on currently profitable deviations binding is also in the self-generating set. The reason
is that, as can be seen from Fig.~\ref{fig:Illustration_ContinuationPayoffs}, the continuation payoff $\gamma_j^*(y)$ that makes the IC constraint
binding has the smallest $\gamma_j^*(y_g^i)=\min \gamma_j^\prime(y_g^i)$ and the largest $\gamma_j^*(y_b^i)=\max \gamma_j^\prime(y_b^i)$. Formally we establish the following Lemma.

\begin{lemma}
Fix a payoff profile $v$ and a discount factor $\delta$. Suppose that $v$ is decomposed by $\besti$. If there are any feasible continuation payoffs
$\gamma^\prime(y_g^i)\in V_{\mu}$ and $\gamma^\prime(y_b^i)\in V_{\mu}$ that satisfy  (\ref{eqn:IC_inactive_simplified}) for all $j\not=i$, there there
exist feasible continuation payoffs $\gamma^*(y_g^i)\in V_{\mu}$ and $\gamma^*(y_b^i)\in V_{\mu}$ such that the IC constraint on currently
profitable deviations in  (\ref{eqn:IC_inactive_simplified}) is binding for all $j\not=i$.
\end{lemma}
\begin{proof}
Given feasible continuation payoffs $\gamma^\prime(y_g^i)\in V_{\mu}$ and $\gamma^\prime(y_b^i)\in V_{\mu}$, we construct $\gamma^*(y_g^i)\in
V_{\mu}$ and $\gamma^*(y_b^i)\in V_{\mu}$ that are feasible and make the IC constraint on currently profitable deviations in
 (\ref{eqn:IC_inactive_simplified}) binding for all $j\not=i$.

Specifically, we set $\gamma_j^*(y_g^i)$ and $\gamma_j^*(y_b^i)$ such that the IC constraint on currently profitable deviations in
 (\ref{eqn:IC_inactive_simplified}) is binding. Such $\gamma_j^*(y_g^i)$ and $\gamma_j^*(y_b^i)$ have the following property:
$\gamma_j^*(y_g^i)\leq\gamma_j^\prime(y_g^i)$ and $\gamma_j^*(y_b^i)\geq\gamma_j^\prime(y_b^i)$ for all $\gamma_j^\prime(y_g^i)$ and
$\gamma_j^\prime(y_b^i)$ that satisfy  (\ref{eqn:IC_inactive_simplified}). We prove this property by contradiction. Suppose that there exist
$\gamma_j^\prime(y_g^i)$ and $\gamma_j^\prime(y_b^i)$ that satisfy  (\ref{eqn:IC_inactive_simplified}) and
$\gamma_j^\prime(y_g^i)=\gamma_j^*(y_g^i)-\Delta$ with $\Delta>0$. Based on the decomposition equality, we have
$$
\gamma_j^\prime(y_b^i) = \gamma_j^*(y_b^i) + \left( \frac{\rho(y_g^i|\besti)}{1-\rho(y_g^i|\besti)} \right) \Delta
$$
We can see that the IC constraint on currently profitable deviations is violated:
\begin{eqnarray}
& & \kappa_{ij}^+ \, \gamma_j^\prime(y_g^i) + (1-\kappa_{ij}^+) \, \gamma_j^\prime(y_b^i) \nonumber \\
&=& \kappa_{ij}^+ \, \gamma_j^*(y_g^i) + (1-\kappa_{ij}^+) \, \gamma_j^*(y_b^i) + \left[ -\kappa_{ij}^+ \, \Delta + (1-\kappa_{ij}^+)
\left(\frac{\rho(y_g^i|\besti)}{1-\rho(y_g^i|\besti)}\right) \Delta \right] \nonumber \\
&=& v_j + (1-\kappa_{ij}^+)\left[ \frac{\rho(y_g^i|\besti)}{1-\rho(y_g^i|\besti)}-\frac{\kappa_{ij}^+}{1-\kappa_{ij}^+}
 \right]  \Delta \nonumber \\
&>& v_j \nonumber
\end{eqnarray}
where the last inequality results from $\kappa_{ij}^+\leq0$. Hence, we have $\gamma_j^*(y_g^i)\leq\gamma_j^\prime(y_g^i)$ and
$\gamma_j^*(y_b^i)\geq\gamma_j^\prime(y_b^i)$ for all $\gamma_j^\prime(y_g^i)$ and $\gamma_j^\prime(y_b^i)$ that satisfy
 (\ref{eqn:IC_inactive_simplified}).

Next, we prove that if $\gamma^\prime(y)\in V_{\mu}$, then $\gamma^*(y)\in V_{\mu}$. To prove $\gamma^*(y)\in V_{\mu}$, we need to show that
$\gamma_j^*(y_g^i)\geq \mu_j$ and $\gamma_j^*(y_b^i)\geq \mu_j$ for all $j\in N$. For $j\not=i$, we have
$\gamma_j^*(y_g^i)\geq\gamma_j^*(y_b^i)\geq\gamma_j^\prime(y_b^i)\geq \mu_j$. For $i$, we have
$$
\gamma_i^*(y_g^i) = \frac{1}{\lambda_i}\left(1-\sum_{j\neq i} \lambda_j \gamma_j^*(y_g^i)\right) \geq \frac{1}{\lambda_i}\left(1-\sum_{j\neq i}
\lambda_j \gamma_j^\prime(y_g^i)\right) = \gamma_i^\prime(y_g^i) \geq \mu_i
$$

This proves the lemma.
\end{proof}

Using this Lemma, we can calculate the continuation payoffs of the inactive player $j\not=i$:
\begin{eqnarray}
\gamma_j(y_g^i) &=& \frac{\big(\frac{1}{\delta}(1-\kappa_{ij}^+) - [1-\rho(y_g^i|\besti)]\big) v_j-(\frac{1}{\delta}-1)(1-\kappa_{ij}^+)  \tilde{v}_j^i}{\rho(y_g^i|\besti)-\kappa_{ij}^+} \nonumber \\
                &=& \frac{v_j}{\delta} - \left(\frac{1-\delta}{\delta}\right) \tilde{v}_j^i + \left(\frac{1-\delta}{\delta}\right) [1-\rho(y_g^i|\besti)] \alpha(i,j), \nonumber \\
\gamma_j(y_b^i) &=& \frac{\left[\rho(y_g^i|\besti)-\frac{1}{\delta}\kappa_{ij}^+\right]  v_j + (\frac{1}{\delta}-1)\kappa_{ij}^+ \,
\tilde{v}_j^i}{\rho(y_g^i|\besti)-\kappa_{ij}^+} \nonumber \\
                &=& \frac{v_j}{\delta} - \left(\frac{1-\delta}{\delta}\right) \tilde{v}_j^i - \left(\frac{1-\delta}{\delta}\right) \rho(y_g^i|\besti)  \alpha(i,j). \nonumber
\end{eqnarray}

The active player's continuation payoffs can be determined based on the inactive players' continuation payoffs since $\gamma(y)\in V$. We calculate
the active player $i$'s continuation payoffs as
\begin{eqnarray}\label{eqn:continuation_payoff_active}
\gamma_i(y_g^i) &=& \frac{v_i}{\delta} - \left(\frac{1-\delta}{\delta}\right) \tilde{v}_i^i - \left(\frac{1-\delta}{\delta}\right)[1-\rho(y_g^i|\besti)] \frac{1}{\lambda_i} \sum_{j\neq i} \lambda_j \alpha(i,j), \nonumber \\
\gamma_i(y_b^i) &=& \frac{v_i}{\delta} - \left(\frac{1-\delta}{\delta}\right) \tilde{v}_i^i + \left(\frac{1-\delta}{\delta}\right) \rho(y_g^i|\besti)
\frac{1}{\lambda_i} \sum_{j\neq i} \lambda_j \alpha(i,j) \nonumber
\end{eqnarray}
Hence, the constraint $v\in\mathscr{B}(V_{\mu};\delta,\besti)$ on discount factor $\delta$ is equivalent to
$$
\gamma(y)\in V_{\mu} \mbox{ for all } y\in Y \Leftrightarrow \gamma_i(y) \geq \mu_i  \mbox{ for all }   i\in N, y\in Y
$$
Since $\kappa_{ij}^+(\mu_j)\leq0$, we have $\gamma_j(y)\geq v_j$ for all $y\in Y$, which means that $\gamma_j(y)\geq \mu_j$ for all $y\in Y$. Hence, we
only need the discount factor to have the property that $\gamma_i(y)\geq \mu_i$ for all $y\in Y$. Since $\gamma_i(y_g^i)<\gamma_i(y_b^i)$, we need
$\gamma_i(y_g^i)\geq \mu_i$, which leads to
$$
\delta \geq \frac{1}{1+\lambda_i(v_i-\mu_i)/\left[\lambda_i(\tilde{v}_i^i-v_i)+\sum_{j\neq i}
\lambda_j\cdot(1-\rho(y_g^i|\besti))\alpha(i,j)\right]}.
$$

Hence, the optimization problem \eqref{eqn:delta_optimization} is equivalent to
\begin{eqnarray}\label{eqn:delta_matrix}
\underline{\delta}(\mu) = \max_{v\in V_{\mu}} \min_{i\in\mathcal{N}} x_i(v)
\end{eqnarray}
where
$$
x_{i}(v)\triangleq \frac{1}{1+\lambda_i(v_i-\mu_i)/\left(\lambda_i(\tilde{v}_i^i-v_i)+\sum_{j\neq i} \lambda_j [1-\rho(y_g^i|\besti)]
\alpha(i,j)\right)}
$$
Since $x_i(v)$ is decreasing in $v_i$, the payoff $v^*$ that maximizes $\min_{i\in\mathcal{N}} x_i(v)$ must satisfy $x_i(v^*)=x_j(v^*)$ for all $i$
and $j$. Now we find the payoff $v^*$ such that $x_i(v^*)=x_j(v^*)$ for all $i$ and $j$.

Define $$z\triangleq \frac{\lambda_i(v_i^*-\mu_i)}{\lambda_i(\tilde{v}_i^i-v_i^*)+\sum\limits_{j\neq i} \lambda_j[1-\rho(y_g^i|\besti)] \alpha(i,j)}$$
Then we have
$$
 \lambda_i(1+z)v_i^* = \lambda_i(\mu_i+z\tilde{v}_i^i) - z \sum_{j\neq i} \lambda_j   [1-\rho(y_g^i|\besti)] \alpha(i,j)
$$
from which it follows that
$$
z = \frac{1 - \sum\limits_{i} \lambda_i\mu_i}{\sum\limits_{i} \left(\lambda_i\tilde{v}_i^i + \sum\limits_{j\neq i} \lambda_j
[1-\rho(y_g^i|\besti)] \alpha(i,j)\right) - 1}
$$
Hence, the minimum discount factor is $\underline{\delta}(\mu) = \frac{1}{1+z}$; substituting the definition of $z$ yields Condition~4.  This completes the proof that these Conditions 1-4 are necessary for $V_\mu$ to be a self-generating set.

It remains to show that these necessary Conditions are also sufficient, which is accomplished in the proof of Theorem 2.  This completes the proof of Theorem 1. $\Box$

\bigskip

\noindent {\bf Proof of Theorem 2 }  In view of the results of APS, it suffices to show  that the algorithm yields a decomposition of each
 target vector $v(t) \in V_\mu$.  The algebra in the proofs of Propositions 2 and 3 shows that Conditions 1 and 2 guarantee that the incentive compatibility constraints are satisfied for the inactive and active players.  The algebra in the proof of Theorem 1 shows that Conditions 3 and 4 taken together guarantee that the continuation payoff $\gamma(y)$ belongs to $V_{\mu}$ for each $y \in Y$.     $\Box$

 \bigskip

\noindent {\bf Proof of Theorem 3 }  Given a parameter $\xi$, the algorithm uses the target vector as the continuation value to compute $N$ indicators; let $\Xi(0)$ be the set of parameters for which no two of these indicators are equal.  For each parameter in $\Xi(0)$, the algorithm computes continuation values following the good signal and the bad signal and then uses each of these continuation values to compute $N$ indicators; let $\Xi(1) \subset \Xi(0)$ be the set of parameters for which no two of these indicators are equal.  Proceeding by induction, we define a decreasing sequence of sets $\Xi(0) \supset \Xi(1) \supset \cdots \supset \Xi(T)$; let $\Xi_T$ be the complement of
$\Xi(T)$.  Notice that the indicators are continuous functions of the parameters so the ordering of the indicators is locally constant provided no two indicators are equal.  Hence for each $\xi \in \Xi(T) = \Xi \setminus \Xi(T)$ then there is a small open neighborhood $Z$ of $\xi$ so that if $\xi' \in Z$ then the strategies $\sigma_{\xi'}, \sigma_\xi$ generate the same ordering of indicators in each of the first $T$ periods. In particular, $\sigma_{\xi'}(h) =  \sigma_\xi(h)$ for each history $h \in \hist^T$; that is, $\xi \to \sigma^T_\xi$ is locally constant on the complement of $\Xi_T$.  It remains only to show that $\Xi_T$ is closed and has measure 0. In fact, $\Xi_T$ is a finite union of lower-dimensional submanifolds; this is a consequence of general facts about semi-algebraic sets and the observation that all the indicators are continuous semi-algebraic functions of the parameters, no two of which coincide on any open set. See \citet*{BCR}, \citet*{BlumeZame}.

\bigskip

\noindent {\bf Proof of Theorem 4 }  Propositions 2, 3 show that Conditions 1, 2 are necessary conditions for the existence of an efficient PPE for {\em any} discount factor..  Suppose therefore that Conditions 1,2 are satisfied.  It is easily checked that the definitions of $\overline{\mu}_1, \overline{\mu}_2$ guarantee that Condition 3 of Theorem 1 are satisfied.  Finally, if $\delta \geq \delta^*$ then Condition 4 of Theorem 1 is also satisfied.  It follows from Theorem 1 that for each $\delta \geq \delta^*$,
$V_{\overline{\mu}}$ is a self-generating set, so every target vector in $V_{\overline{\mu}}$ can be achieved in a PPE.  Hence $E(\delta) \supset V_{\overline{\mu}}$ for every $\delta \in [\delta^*, 1)$.
To see that $V_{\overline{\mu}} = E(\delta)$ for every $\delta \in [\delta^*, 1)$, simply note that for each $\delta$ the set $E(\delta)$ is closed and convex, hence an interval, hence of the form $V_\mu$ for some $\mu$.  However, Condition 3 of Theorem 1 guarantees that $\mu \geq \overline{\mu}$ which completes the proof.  $\Box$

 \end{document}